\documentclass[a4paper,twocolumn,superscriptaddress,11pt]{quantumarticle}

\usepackage{hyperref}
\usepackage{amsmath,amsthm,amssymb}
\usepackage{xspace,enumerate,color,epsfig}
\usepackage{graphicx}
\graphicspath{{.}{./figures/}}
\usepackage{marvosym}
\usepackage{bm}

\usepackage{tikzfig}
\usepackage{stmaryrd}
\usepackage{docmute}
\usepackage{keycommand}

\newkeycommand{\boxmap}[style=small box][1]{\,\tikz{\node[style=\commandkey{style}] (x) {$#1$};\draw(0,-1.25)--(x)--(0,1.25);}\,}
\newkeycommand{\dboxmap}[style=dbox][1]{\,\tikz{\node[style=\commandkey{style}] (x) {$#1$};\draw[boldedge] (0,-1.25)--(x)--(0,1.25);}\,}

\newkeycommand{\wirelabel}[style=white label]{\,\tikz{\node[style=\commandkey{style}]{};}\,\xspace}

\newkeycommand{\pointmap}[style=point][1]{\,\tikz{\node[style=\commandkey{style}] (x) at (0,-0.3) {$#1$};\node [style=none] (3) at (0, 0.9) {};\node [style=none] (3) at (0, -0.9) {};\draw (x)--(0,0.7);}\,}

\newkeycommand{\point}[style=point,estyle=][1]{\,\tikz{\node[style=\commandkey{style}] (x) at (0,-0.05) {$#1$};\draw [\commandkey{estyle}] (x)--(0,1.05);}\,}

\newkeycommand{\copointmap}[style=copoint][1]{\,\tikz{\node[style=\commandkey{style}] (x) at (0,0.3) {$#1$};\draw (x)--(0,-0.7);}\,}

\newkeycommand{\dpoint}[style=point][1]{\,\tikz{\node[style=\commandkey{style},doubled] (x) at (0,-0.3) {$#1$};\draw[boldedge] (x)--(0,0.7);}\,}

\newkeycommand{\kpoint}[style=kpoint,estyle=][1]{\,\tikz{\node[style=\commandkey{style}] (x) at (0,-0.05) {$#1$};\draw [\commandkey{estyle}] (x)--(0,1.05);}\,}

\newkeycommand{\kpointgray}[style=gray kpoint,estyle=][1]{\,\tikz{\node[style=\commandkey{style}] (x) at (0,-0.05) {$#1$};\draw [\commandkey{estyle}] (x)--(0,1.05);}\,}

\newkeycommand{\typedkpoint}[style=kpoint,edgestyle=][2]{%
\,\begin{tikzpicture}
  \begin{pgfonlayer}{nodelayer}
    \node [style=none] (0) at (0, 1) {};
    \node [style=\commandkey{style}] (1) at (0, -0.75) {$#2$};
    \node [style=right label] (2) at (0.25, 0.5) {$#1$};
  \end{pgfonlayer}
  \begin{pgfonlayer}{edgelayer}
    \draw [\commandkey{edgestyle}] (1) to (0.center);
  \end{pgfonlayer}
\end{tikzpicture}\,}

\newkeycommand{\typedkpointdag}[style=kpoint adjoint,edgestyle=][2]{%
\,\begin{tikzpicture}
  \begin{pgfonlayer}{nodelayer}
    \node [style=none] (0) at (0, -1) {};
    \node [style=\commandkey{style}] (1) at (0, 0.75) {$#2$};
    \node [style=right label] (2) at (0.25, -0.5) {$#1$};
  \end{pgfonlayer}
  \begin{pgfonlayer}{edgelayer}
    \draw [\commandkey{edgestyle}] (0.center) to (1);
  \end{pgfonlayer}
\end{tikzpicture}\,}

\newkeycommand{\bistate}[style=kpoint][1]{\,\begin{tikzpicture}
  \begin{pgfonlayer}{nodelayer}
    \node [style=\commandkey{style}, minimum width=1 cm, inner sep=2pt] (0) at (0, -0.25) {$#1$};
    \node [style=none] (1) at (-0.75, 0) {};
    \node [style=none] (2) at (0.75, 0) {};
    \node [style=none] (3) at (-0.75, 0.88) {};
    \node [style=none] (4) at (0.75, 0.88) {};
  \end{pgfonlayer}
  \begin{pgfonlayer}{edgelayer}
    \draw (1.center) to (3.center);
    \draw (2.center) to (4.center);
  \end{pgfonlayer}
\end{tikzpicture}\,}

\newkeycommand{\bistatebig}[style=kpoint][1]{\,\begin{tikzpicture}
  \begin{pgfonlayer}{nodelayer}
    \node [style=\commandkey{style}, minimum width=, minimum width=1.26 cm, inner sep=2pt] (0) at (0, -0.25) {$#1$};
    \node [style=none] (1) at (-1, 0) {};
    \node [style=none] (2) at (1, 0) {};
    \node [style=none] (3) at (-1, 0.88) {};
    \node [style=none] (4) at (1, 0.88) {};
  \end{pgfonlayer}
  \begin{pgfonlayer}{edgelayer}
    \draw (1.center) to (3.center);
    \draw (2.center) to (4.center);
  \end{pgfonlayer}
\end{tikzpicture}\,}

\newkeycommand{\dbistate}[style=dkpoint][1]{\,\begin{tikzpicture}
  \begin{pgfonlayer}{nodelayer}
    \node [style=\commandkey{style}, minimum width=1 cm, inner sep=2pt] (0) at (0, -0.25) {$#1$};
    \node [style=none] (1) at (-0.75, 0) {};
    \node [style=none] (2) at (0.75, 0) {};
    \node [style=none] (3) at (-0.75, 1.0) {};
    \node [style=none] (4) at (0.75, 1.0) {};
  \end{pgfonlayer}
  \begin{pgfonlayer}{edgelayer}
    \draw [boldedge] (1.center) to (3.center);
    \draw [boldedge] (2.center) to (4.center);
  \end{pgfonlayer}
\end{tikzpicture}\,}

\newkeycommand{\bistatebraket}[style1=kpoint,style2=kpointadj][2]{\,\begin{tikzpicture}
  \begin{pgfonlayer}{nodelayer}
    \node [style=\commandkey{style1}, minimum width=1 cm, inner sep=2pt] (0) at (0, -0.75) {$#1$};
    \node [style=none] (1) at (-0.75, -0.5) {};
    \node [style=none] (2) at (0.75, -0.5) {};
    \node [style=none] (3) at (-0.75, 0.5) {};
    \node [style=none] (4) at (0.75, 0.5) {};
    \node [style=\commandkey{style2}, minimum width=1 cm, inner sep=2pt] (5) at (0, 0.75) {$#2$};
  \end{pgfonlayer}
  \begin{pgfonlayer}{edgelayer}
    \draw (1.center) to (3.center);
    \draw (2.center) to (4.center);
  \end{pgfonlayer}
\end{tikzpicture}\,}

\newkeycommand{\dkpoint}[style=dkpoint][1]{\,\tikz{\node[style=\commandkey{style}] (x) at (0,-0.1) {$#1$};\draw[boldedge] (x)--(0,1);}\,}
\newkeycommand{\dkpointadj}[style=dkpointadj][1]{\,\tikz{\node[style=\commandkey{style}] (x) at (0,0.1) {$#1$};\draw[boldedge] (x)--(0,-1);}\,}
\newkeycommand{\dkpointtrans}[style=dkpointtrans][1]{\,\tikz{\node[style=\commandkey{style}] (x) at (0,0.1) {$#1$};\draw[boldedge] (x)--(0,-1);}\,}
\newkeycommand{\dkpointconj}[style=dkpointconj][1]{\,\tikz{\node[style=\commandkey{style}] (x) at (0,-0.1) {$#1$};\draw[boldedge] (x)--(0,1);}\,}

\newcommand{\trace}{\,\begin{tikzpicture}
  \begin{pgfonlayer}{nodelayer}
    \node [style=none] (0) at (0, -0.60) {};
    \node [style=upground] (1) at (0, 0.40) {};
  \end{pgfonlayer}
  \begin{pgfonlayer}{edgelayer}
    \draw [style=boldedge] (0.center) to (1);
  \end{pgfonlayer}
\end{tikzpicture}\,}

\newcommand\discard\trace

\newcommand\oneoverD{\ensuremath{{\textstyle{1\over{D}}}}\xspace}

\newkeycommand{\pointketbra}[style1=copoint,style2=point][2]{\,%
\begin{tikzpicture}
  \begin{pgfonlayer}{nodelayer}
    \node [style=none] (0) at (0, -1.75) {};
    \node [style=\commandkey{style1}] (1) at (0, -0.75) {$#1$};
    \node [style=\commandkey{style2}] (2) at (0, 0.75) {$#2$};
    \node [style=none] (3) at (0, 1.75) {};
  \end{pgfonlayer}
  \begin{pgfonlayer}{edgelayer}
    \draw (0.center) to (1);
    \draw (2) to (3.center);
  \end{pgfonlayer}
\end{tikzpicture}\,}

\newkeycommand{\twocopointketbra}[style1=copoint,style2=copoint,style3=point][3]{\,%
\begin{tikzpicture}
  \begin{pgfonlayer}{nodelayer}
    \node [style=none] (0) at (-0.75, -1.75) {};
    \node [style=none] (1) at (0.75, -1.75) {};
    \node [style=\commandkey{style1}] (2) at (-0.75, -0.75) {$#1$};
    \node [style=\commandkey{style2}] (3) at (0.75, -0.75) {$#2$};
    \node [style=\commandkey{style3}] (4) at (0, 0.75) {$#3$};
    \node [style=none] (5) at (0, 1.75) {};
  \end{pgfonlayer}
  \begin{pgfonlayer}{edgelayer}
    \draw (0.center) to (2);
    \draw (1.center) to (3);
    \draw (4) to (5.center);
  \end{pgfonlayer}
\end{tikzpicture}\,}

\newkeycommand{\twopointketbra}[style1=copoint,style2=point,style3=point][3]{\,%
\begin{tikzpicture}
  \begin{pgfonlayer}{nodelayer}
    \node [style=none] (0) at (0, -1.75) {};
    \node [style=none] (1) at (-0.75, 1.75) {};
    \node [style=\commandkey{style1}] (2) at (0, -0.75) {$#1$};
    \node [style=\commandkey{style2}] (3) at (-0.75, 0.75) {$#2$};
    \node [style=\commandkey{style3}] (4) at (0.75, 0.75) {$#3$};
    \node [style=none] (5) at (0.75, 1.75) {};
  \end{pgfonlayer}
  \begin{pgfonlayer}{edgelayer}
    \draw (0.center) to (2);
    \draw (1.center) to (3);
    \draw (4) to (5.center);
  \end{pgfonlayer}
\end{tikzpicture}\,}

\newkeycommand{\pointbraket}[style1=point,style2=copoint,estyle=][2]{\,%
\begin{tikzpicture}
  \begin{pgfonlayer}{nodelayer}
    \node [style=\commandkey{style1}] (0) at (0, -0.625) {$#1$};
    \node [style=\commandkey{style2}] (1) at (0, 0.625) {$#2$};
  \end{pgfonlayer}
  \begin{pgfonlayer}{edgelayer}
    \draw [\commandkey{estyle}] (0) to (1);
  \end{pgfonlayer}
\end{tikzpicture}\,}

\newkeycommand{\kpointketbra}[style1=kpointdag,style2=kpoint,estyle=][2]{\,%
\begin{tikzpicture}
  \begin{pgfonlayer}{nodelayer}
    \node [style=none] (0) at (0, -1.9) {};
    \node [style=\commandkey{style1}] (1) at (0, -0.95) {$#1$};
    \node [style=\commandkey{style2}] (2) at (0, 0.95) {$#2$};
    \node [style=none] (3) at (0, 1.9) {};
  \end{pgfonlayer}
  \begin{pgfonlayer}{edgelayer}
    \draw [\commandkey{estyle}] (0.center) to (1);
    \draw [\commandkey{estyle}] (2) to (3.center);
  \end{pgfonlayer}
\end{tikzpicture}\,}

\newkeycommand{\kpointmap}[style1=kpoint,style2=map][2]{\,%
\begin{tikzpicture}
  \begin{pgfonlayer}{nodelayer}
    \node [style=\commandkey{style1}] (0) at (0, -1.1) {$#1$};
    \node [style=\commandkey{style2}] (1) at (0, 0.5) {$#2$};
    \node [style=none] (2) at (0, 1.75) {};
  \end{pgfonlayer}
  \begin{pgfonlayer}{edgelayer}
    \draw (0) to (1);
    \draw (1) to (2.center);
  \end{pgfonlayer}
\end{tikzpicture}%
\,}

\newkeycommand{\dkpointmap}[style1=dkpoint,style2=dmap][2]{\,%
\begin{tikzpicture}
  \begin{pgfonlayer}{nodelayer}
    \node [style=\commandkey{style1}] (0) at (0, -1.2) {$#1$};
    \node [style=\commandkey{style2}] (1) at (0, 0.4) {$#2$};
    \node [style=none] (2) at (0, 1.7) {};
  \end{pgfonlayer}
  \begin{pgfonlayer}{edgelayer}
    \draw[style=boldedge] (0) to (1);
    \draw[style=boldedge] (1) to (2.center);
  \end{pgfonlayer}
\end{tikzpicture}%
\,}

\newkeycommand{\dkpointmapx}[style1=dkpoint,style2=dmap][2]{\,%
\begin{tikzpicture}
  \begin{pgfonlayer}{nodelayer}
    \node [style=\commandkey{style1}] (0) at (0, -1.4) {$#1$};
    \node [style=\commandkey{style2}] (1) at (0, 0.6) {$#2$};
    \node [style=none] (2) at (0, 1.9) {};
  \end{pgfonlayer}
  \begin{pgfonlayer}{edgelayer}
    \draw[style=boldedge] (0) to (1);
    \draw[style=boldedge] (1) to (2.center);
  \end{pgfonlayer}
\end{tikzpicture}%
\,}

\newkeycommand{\kpointsandwichmap}[style1=kpoint,style2=map,style3=kpointdag][3]{\,%
\begin{tikzpicture}
  \begin{pgfonlayer}{nodelayer}
    \node [style=\commandkey{style1}] (0) at (0, -1.5) {$#1$};
    \node [style=\commandkey{style2}] (1) at (0, 0) {$#2$};
    \node [style=\commandkey{style3}] (2) at (0, 1.5) {$#3$};
  \end{pgfonlayer}
  \begin{pgfonlayer}{edgelayer}
    \draw (0) to (1);
    \draw (1) to (2);
  \end{pgfonlayer}
\end{tikzpicture}%
}

\newkeycommand{\sandwichtwo}[style1=point,style2=map,style3=map,style4=copoint][4]{\,%
\begin{tikzpicture}
  \begin{pgfonlayer}{nodelayer}
    \node [style=\commandkey{style2}] (0) at (0, -1) {$#2$};
    \node [style=\commandkey{style3}] (1) at (0, 1) {$#3$};
    \node [style=\commandkey{style1}] (2) at (0, -2.6) {$#1$};
    \node [style=\commandkey{style4}] (3) at (0, 2.6) {$#4$};
  \end{pgfonlayer}
  \begin{pgfonlayer}{edgelayer}
    \draw (2) to (0);
    \draw (0) to (1);
    \draw (1) to (3);
  \end{pgfonlayer}
\end{tikzpicture}\,}

\newkeycommand{\longbraket}[style1=point,style2=copoint][2]{\,% 
\begin{tikzpicture}
  \begin{pgfonlayer}{nodelayer}
    \node [style=\commandkey{style1}] (0) at (0, -2.6) {$#1$};
    \node [style=\commandkey{style2}] (1) at (0, 2.6) {$#2$};
  \end{pgfonlayer}
  \begin{pgfonlayer}{edgelayer}
    \draw (0) to (1);
  \end{pgfonlayer}
\end{tikzpicture}\,}

\newkeycommand{\onb}[style=point]{\ensuremath{\{\,\tikz{\node[style=\commandkey{style}] (x) at (0,-0.3) {$j$};\draw (x)--(0,0.7);}\,\}}\xspace}

\newkeycommand{\phase}[style=white phase dot][1]{\,\begin{tikzpicture}
    \begin{pgfonlayer}{nodelayer}
        \node [style=none] (0) at (0, 1) {};
        \node [style=\commandkey{style}] (2) at (0, -0) {$#1$}; 
        \node [style=none] (3) at (0, -1) {};
    \end{pgfonlayer}
    \begin{pgfonlayer}{edgelayer}
        \draw (2) to (0.center);
        \draw (3.center) to (2);
    \end{pgfonlayer}
\end{tikzpicture}\,}

\newkeycommand{\dphase}[style=white phase ddot][1]{\,\begin{tikzpicture}
    \begin{pgfonlayer}{nodelayer}
        \node [style=none] (0) at (0, 1.25) {};
        \node [style=\commandkey{style}] (2) at (0, -0) {$#1$};
        \node [style=none] (3) at (0, -1.25) {};
    \end{pgfonlayer}
    \begin{pgfonlayer}{edgelayer}
        \draw [boldedge] (2) to (0.center);
        \draw [boldedge] (3.center) to (2);
    \end{pgfonlayer}
\end{tikzpicture}\,}

\newkeycommand{\dphasepoint}[style=white phase ddot][1]{\,\begin{tikzpicture}[yshift=-3mm]
  \begin{pgfonlayer}{nodelayer}
    \node [style=none] (0) at (0, 1) {};
    \node [style=\commandkey{style}] (1) at (0, 0) {$#1$};
  \end{pgfonlayer}
  \begin{pgfonlayer}{edgelayer}
    \draw [style=boldedge] (0.center) to (1);
  \end{pgfonlayer}
\end{tikzpicture}\,}

\newkeycommand{\dphasepointgray}[style=gray phase ddot][1]{\,\begin{tikzpicture}[yshift=-3mm]
  \begin{pgfonlayer}{nodelayer}
    \node [style=none] (0) at (0, 1) {};
    \node [style=\commandkey{style}] (1) at (0, 0) {$#1$};
  \end{pgfonlayer}
  \begin{pgfonlayer}{edgelayer}
    \draw [style=boldedge] (0.center) to (1);
  \end{pgfonlayer}
\end{tikzpicture}\,}

\newkeycommand{\dphasecopoint}[style=white phase ddot][1]{\,\begin{tikzpicture}[yshift=3mm,yscale=-1]
  \begin{pgfonlayer}{nodelayer}
    \node [style=none] (0) at (0, 1) {};
    \node [style=\commandkey{style}] (1) at (0, 0) {$#1$};
  \end{pgfonlayer}
  \begin{pgfonlayer}{edgelayer}
    \draw [style=boldedge] (0.center) to (1);
  \end{pgfonlayer}
\end{tikzpicture}\,}

\newkeycommand{\dphasepointsm}[style=white phase ddot][1]{\,\begin{tikzpicture}[yshift=-3mm]
  \begin{pgfonlayer}{nodelayer}
    \node [style=none] (0) at (0, 1) {};
    \node [style=\commandkey{style}] (1) at (0, 0) {\footnotesize\!$#1$\!};
  \end{pgfonlayer}
  \begin{pgfonlayer}{edgelayer}
    \draw [style=boldedge] (0.center) to (1);
  \end{pgfonlayer}
\end{tikzpicture}\,}

\newkeycommand{\phasepoint}[style=white phase dot][1]{\,\begin{tikzpicture}[yshift=-3mm]
  \begin{pgfonlayer}{nodelayer}
    \node [style=none] (0) at (0, 1) {};
    \node [style=\commandkey{style}] (1) at (0, 0) {$#1$};
  \end{pgfonlayer}
  \begin{pgfonlayer}{edgelayer}
    \draw (0.center) to (1);
  \end{pgfonlayer}
\end{tikzpicture}\,}

\newkeycommand{\phasecopoint}[style=white phase dot][1]{\,\begin{tikzpicture}[yshift=3mm]
  \begin{pgfonlayer}{nodelayer}
    \node [style=none] (0) at (0, -1) {};
    \node [style=\commandkey{style}] (1) at (0, 0) {$#1$};
  \end{pgfonlayer}
  \begin{pgfonlayer}{edgelayer}
    \draw (0.center) to (1);
  \end{pgfonlayer}
\end{tikzpicture}\,}

\newkeycommand{\kpointbraket}[style1=kpoint,style2=kpoint adjoint,estyle=][2]{\,%
\begin{tikzpicture}
  \begin{pgfonlayer}{nodelayer}
    \node [style=\commandkey{style1}] (0) at (0, -0.735) {$#1$};
    \node [style=\commandkey{style2}] (1) at (0, 0.735) {$#2$};
  \end{pgfonlayer}
  \begin{pgfonlayer}{edgelayer}
    \draw [\commandkey{estyle}] (0) to (1);
  \end{pgfonlayer}
\end{tikzpicture}\,}

\newkeycommand{\kpointbraketx}[style1=kpoint,style2=kpoint adjoint,estyle=][2]{\,%
\begin{tikzpicture}
  \begin{pgfonlayer}{nodelayer}
    \node [style=\commandkey{style1}] (0) at (0, -0.885) {$#1$};
    \node [style=\commandkey{style2}] (1) at (0, 0.885) {$#2$};
  \end{pgfonlayer}
  \begin{pgfonlayer}{edgelayer}
    \draw [\commandkey{estyle}] (0) to (1);
  \end{pgfonlayer}
\end{tikzpicture}\,}

\newcommand{\bra}[1]{\ensuremath{\left\langle #1 \right|}}
\newcommand{\ket}[1]{\ensuremath{\left|  #1 \right\rangle}}

\newcommand{\braket}[2]{\ensuremath{\langle#1|#2\rangle}}
\newcommand{\ketbra}[2]{\ensuremath{\ket{#1}\!\bra{#2}}}

\tikzstyle{cdiag}=[matrix of math nodes, row sep=3em, column sep=3em, text height=1.5ex, text depth=0.25ex,inner sep=0.5em]
\tikzstyle{arrow above}=[transform canvas={yshift=0.5ex}]
\tikzstyle{arrow below}=[transform canvas={yshift=-0.5ex}]

\newcommand{\Tr}{\ensuremath{\textrm{Tr}\,}\xspace}
\newcommand\cb{\mathrm{cb}}
\newcommand\N{\mathbb{N}}
\newcommand\id{\mathrm{id}}

\newcommand{\noiseconst}{\ensuremath{N}} %Constant in final Noise proof. 

\newcommand{\ehilb}{\ensuremath{E}} %Eve's `environment' Hilbert space

\newcommand{\ahilb}{\ensuremath{H}} %Alice and Bob's Hilbert space

\newcommand{\evepure}{\ensuremath{V}} %Operator from purifying Eve's channel

\tikzstyle{env}=[copoint,regular polygon rotate=0,minimum width=0.2cm, fill=black]

\tikzstyle{probs}=[shape=semicircle,fill=white,draw=black,shape border rotate=180,minimum width=1.2cm]

\tikzstyle{every picture}=[baseline=-0.25em,scale=0.5]
\tikzstyle{dotpic}=[] % for backwards-compatibility
\tikzstyle{diredges}=[every to/.style={diredge}]
\tikzstyle{math matrix}=[matrix of math nodes,left delimiter=(,right delimiter=),inner sep=2pt,column sep=1em,row sep=0.5em,nodes={inner sep=0pt},text height=1.5ex, text depth=0.25ex]

\tikzstyle{gs edge}=[]
\tikzstyle{gs double edge}=[double,shorten <=-1mm,shorten >=-1mm,double distance=1pt]

\tikzstyle{inline text}=[text height=1.75ex, text depth=0.3ex,yshift=0mm]
\tikzstyle{label}=[font=\tiny,text height=1ex, text depth=0.15ex]
\tikzstyle{left label}=[label,anchor=east,xshift=1.5mm]
\tikzstyle{right label}=[label,anchor=west,xshift=-1.5mm]

\tikzstyle{braceedge}=[decorate,decoration={brace,amplitude=2mm,raise=-1mm}]
\tikzstyle{small braceedge}=[decorate,decoration={brace,amplitude=1mm,raise=-1mm}]

\tikzstyle{doubled}=[line width=1.6pt] % set the line width for all doubled (quantum) maps/wires
\tikzstyle{boldedge}=[doubled,shorten <=-0.17mm,shorten >=-0.17mm]
\tikzstyle{boldedgegray}=[doubled,gray,shorten <=-0.17mm,shorten >=-0.17mm]
\tikzstyle{singleedgegray}=[gray]%,shorten <=-0.1mm,shorten >=-0.1mm]

\tikzstyle{semidoubled}=[line width=1.4pt] % set the line width for all doubled (quantum) maps/wires
\tikzstyle{semiboldedgegray}=[semidoubled,gray,shorten <=-0.17mm,shorten >=-0.17mm]

\tikzstyle{boxedge}=[semiboldedgegray]

\tikzstyle{boldedgedashed}=[very thick,dashed,shorten <=-0.17mm,shorten >=-0.17mm]
\tikzstyle{vboldedgedashed}=[doubled,dashed,shorten <=-0.17mm,shorten >=-0.17mm]
\tikzstyle{left hook arrow}=[left hook-latex]
\tikzstyle{right hook arrow}=[right hook-latex]
\tikzstyle{sembracket}=[line width=0.5pt,shorten <=-0.07mm,shorten >=-0.07mm]

\tikzstyle{causal edge}=[->,thick,gray]
\tikzstyle{causal nondir}=[thick,gray]
\tikzstyle{timeline}=[thick,gray, dashed]

\tikzstyle{cedge}=[<->,thick,gray!70!white]

\tikzstyle{empty diagram}=[draw=gray!40!white,dashed,shape=rectangle,minimum width=1cm,minimum height=1cm]
\tikzstyle{empty diagram small}=[draw=gray!50!white,dashed,shape=rectangle,minimum width=0.6cm,minimum height=0.5cm]

\tikzstyle{dot}=[inner sep=0mm,minimum width=2mm,minimum height=2mm,draw,shape=circle]  
\tikzstyle{Wsquare}=[white dot, shape=regular polygon, rounded corners=0.8 mm, minimum size=3.3 mm, regular polygon sides=3, outer sep=-0.2mm]
\tikzstyle{Wsquareadj}=[white dot, shape=regular polygon, rounded corners=0.8 mm, minimum size=3.3 mm, regular polygon sides=3, outer sep=-0.2mm, regular polygon rotate=180]
\tikzstyle{ddot}=[inner sep=0mm, doubled, minimum width=2.5mm,minimum height=2.5mm,draw,shape=circle]

\tikzstyle{black dot}=[dot,fill=black]
\tikzstyle{white dot}=[dot,fill=white,,text depth=-0.2mm]
\tikzstyle{white Wsquare}=[Wsquare,fill=white,,text depth=-0.2mm]
\tikzstyle{white Wsquareadj}=[Wsquareadj,fill=white,,text depth=-0.2mm]
\tikzstyle{green dot}=[white dot] % for backwards-compatibility
\tikzstyle{gray dot}=[dot,fill=gray!40!white,,text depth=-0.2mm]
\tikzstyle{red dot}=[gray dot] % for backwards-compatibility

\tikzstyle{Z}=[white dot]
\tikzstyle{X}=[gray dot]
\tikzstyle{simple}=[]

\tikzstyle{black ddot}=[ddot,fill=black]
\tikzstyle{white ddot}=[ddot,fill=white]
\tikzstyle{gray ddot}=[ddot,fill=gray!40!white]

\tikzstyle{gray edge}=[gray!60!white]

\tikzstyle{small dot}=[inner sep=0.5mm,minimum width=0pt,minimum height=0pt,draw,shape=circle]

\tikzstyle{small black dot}=[small dot,fill=black]
\tikzstyle{small white dot}=[small dot,fill=white]
\tikzstyle{small gray dot}=[small dot,fill=gray!40!white]

\tikzstyle{mbqc dot}=[small black dot]
\tikzstyle{mbqc input dot}=[small white dot]
\tikzstyle{mbqc output dot}=[small gray dot]

\tikzstyle{causal dot}=[inner sep=0.4mm,minimum width=0pt,minimum height=0pt,draw=white,shape=circle,fill=gray!40!white]

\tikzstyle{phase dimensions}=[minimum size=5mm,font=\footnotesize,rectangle,rounded corners=2mm,inner sep=0.2mm,outer sep=-2mm,scale=0.8]
\tikzstyle{dphase dimensions}=[minimum size=5mm,font=\footnotesize,rectangle,rounded corners=2.5mm,inner sep=0.2mm,outer sep=-2mm]
\tikzstyle{white phase dot}=[dot,fill=white,phase dimensions]
\tikzstyle{white phase ddot}=[ddot,fill=white,dphase dimensions]

\tikzstyle{white rect ddot}=[draw=black,fill=white,doubled,minimum size=5mm,font=\footnotesize,rectangle,rounded corners=2.5mm,inner sep=0.2mm]
\tikzstyle{gray rect ddot}=[draw=black,fill=gray!40!white,doubled,minimum size=6mm,font=\footnotesize,rectangle,rounded corners=3mm]

\tikzstyle{gray phase dot}=[dot,fill=gray!40!white,phase dimensions]
\tikzstyle{gray phase ddot}=[ddot,fill=gray!40!white,dphase dimensions]
\tikzstyle{grey phase dot}=[gray phase dot]
\tikzstyle{grey phase ddot}=[gray phase ddot]

\tikzstyle{small phase dimensions}=[minimum size=4mm,font=\tiny,rectangle,rounded corners=2mm,inner sep=0.2mm,outer sep=-2mm]
\tikzstyle{small dphase dimensions}=[minimum size=4mm,font=\tiny,rectangle,rounded corners=2mm,inner sep=0.2mm,outer sep=-2mm]

\tikzstyle{small gray phase dot}=[dot,fill=gray!40!white,small phase dimensions]
\tikzstyle{small gray phase ddot}=[ddot,fill=gray!40!white,small dphase dimensions]

\tikzstyle{small map}=[draw,shape=rectangle,minimum height=4mm,minimum width=4mm,fill=white]

\tikzstyle{cnot}=[fill=white,shape=circle,inner sep=-1.4pt]

\tikzstyle{asym hadamard}=[fill=white,draw,shape=NEbox,inner sep=0.6mm,font=\footnotesize,minimum height=4mm]
\tikzstyle{asym hadamard conj}=[fill=white,draw,shape=NWbox,inner sep=0.6mm,font=\footnotesize,minimum height=4mm]
\tikzstyle{asym hadamard dag}=[fill=white,draw,shape=SEbox,inner sep=0.6mm,font=\footnotesize,minimum height=4mm]

\tikzstyle{hadamard}=[fill=white,draw,inner sep=0.6mm,font=\footnotesize,minimum height=4mm,minimum width=4mm]
\tikzstyle{small hadamard}=[fill=white,draw,inner sep=0.6mm,minimum height=1.5mm,minimum width=1.5mm]
\tikzstyle{small hadamard rotate}=[small hadamard,rotate=45]
\tikzstyle{dhadamard}=[hadamard,doubled]
\tikzstyle{small dhadamard}=[small hadamard,doubled]
\tikzstyle{small dhadamard rotate}=[small hadamard rotate,doubled]
\tikzstyle{antipode}=[white dot,inner sep=0.3mm,font=\footnotesize]

\tikzstyle{scalar}=[diamond,draw,inner sep=0.5pt,font=\small]
\tikzstyle{dscalar}=[diamond,doubled, draw,inner sep=0.5pt,font=\small]

\tikzstyle{small box}=[rectangle,inline text,fill=white,draw,minimum height=5mm,yshift=-0.5mm,minimum width=5mm,font=\small]
\tikzstyle{small gray box}=[small box,fill=gray!30]
\tikzstyle{medium box}=[rectangle,inline text,fill=white,draw,minimum height=5mm,yshift=-0.5mm,minimum width=10mm,font=\small]
\tikzstyle{square box}=[small box] % for backwards-compatibility
\tikzstyle{medium gray box}=[small box,fill=gray!30]
\tikzstyle{semilarge box}=[rectangle,inline text,fill=white,draw,minimum height=5mm,yshift=-0.5mm,minimum width=12.5mm,font=\small]
\tikzstyle{large box}=[rectangle,inline text,fill=white,draw,minimum height=5mm,yshift=-0.5mm,minimum width=15mm,font=\small]
\tikzstyle{large gray box}=[small box,fill=gray!30]

\tikzstyle{Bayes box}=[rectangle,fill=black,draw, minimum height=3mm, minimum width=3mm]

\tikzstyle{gray square point}=[small box,fill=gray!50]

\tikzstyle{dphase box white}=[dhadamard]
\tikzstyle{dphase box gray}=[dhadamard,fill=gray!50!white]
\tikzstyle{phase box white}=[hadamard]
\tikzstyle{phase box gray}=[hadamard,fill=gray!50!white]

\tikzstyle{point}=[regular polygon,regular polygon sides=3,draw,scale=0.75,inner sep=-0.5pt,minimum width=9mm,fill=white,regular polygon rotate=180]
\tikzstyle{copoint}=[regular polygon,regular polygon sides=3,draw,scale=0.75,inner sep=-0.5pt,minimum width=9mm,fill=white]
\tikzstyle{dpoint}=[point,doubled]
\tikzstyle{dcopoint}=[copoint,doubled]

\tikzstyle{wide copoint}=[fill=white,draw,shape=isosceles triangle,shape border rotate=90,isosceles triangle stretches=true,inner sep=0pt,minimum width=1.5cm,minimum height=6.12mm]
\tikzstyle{wide point}=[fill=white,draw,shape=isosceles triangle,shape border rotate=-90,isosceles triangle stretches=true,inner sep=0pt,minimum width=1.5cm,minimum height=6.12mm,yshift=-0.0mm]
\tikzstyle{wide point plus}=[fill=white,draw,shape=isosceles triangle,shape border rotate=-90,isosceles triangle stretches=true,inner sep=0pt,minimum width=1.74cm,minimum height=7mm,yshift=-0.0mm]

\tikzstyle{wide dpoint}=[fill=white,doubled,draw,shape=isosceles triangle,shape border rotate=-90,isosceles triangle stretches=true,inner sep=0pt,minimum width=1.5cm,minimum height=6.12mm,yshift=-0.0mm]

\tikzstyle{tinypoint}=[regular polygon,regular polygon sides=3,draw,scale=0.55,inner sep=-0.15pt,minimum width=6mm,fill=white,regular polygon rotate=180] 

\tikzstyle{white point}=[point]
\tikzstyle{white dpoint}=[dpoint]
\tikzstyle{green point}=[white point] % for backwards-compatibility
\tikzstyle{white copoint}=[copoint]
\tikzstyle{gray point}=[point,fill=gray!40!white]
\tikzstyle{gray dpoint}=[gray point,doubled]
\tikzstyle{red point}=[gray point] % for backwards-compatibility
\tikzstyle{gray copoint}=[copoint,fill=gray!40!white]
\tikzstyle{gray dcopoint}=[gray copoint,doubled]

\tikzstyle{white point guide}=[regular polygon,regular polygon sides=3,font=\scriptsize,draw,scale=0.65,inner sep=-0.5pt,minimum width=9mm,fill=white,regular polygon rotate=180]

\tikzstyle{black point}=[point,fill=black,font=\color{white}]
\tikzstyle{black copoint}=[copoint,fill=black,font=\color{white}]

\tikzstyle{tiny gray point}=[tinypoint,fill=gray!40!white]

\tikzstyle{diredge}=[->]
\tikzstyle{ddiredge}=[<->]
\tikzstyle{rdiredge}=[<-]
\tikzstyle{thickdiredge}=[->, very thick]
\tikzstyle{pointer edge}=[->,very thick,gray]
\tikzstyle{pointer edge part}=[very thick,gray]
\tikzstyle{dashed edge}=[dashed]
\tikzstyle{thick dashed edge}=[very thick,dashed]
\tikzstyle{thick gray dashed edge}=[thick dashed edge,gray!40]
\tikzstyle{thick map edge}=[very thick,|->]

\makeatletter
\newcommand{\boxshape}[3]{%
\pgfdeclareshape{#1}{
\inheritsavedanchors[from=rectangle] % this is nearly a rectangle
\inheritanchorborder[from=rectangle]
\inheritanchor[from=rectangle]{center}
\inheritanchor[from=rectangle]{north}
\inheritanchor[from=rectangle]{south}
\inheritanchor[from=rectangle]{west}
\inheritanchor[from=rectangle]{east}
\backgroundpath{% this is new
\southwest \pgf@xa=\pgf@x \pgf@ya=\pgf@y
\northeast \pgf@xb=\pgf@x \pgf@yb=\pgf@y

\@tempdima=#2
\@tempdimb=#3

\pgfpathmoveto{\pgfpoint{\pgf@xa - 5pt + \@tempdima}{\pgf@ya}}
\pgfpathlineto{\pgfpoint{\pgf@xa - 5pt - \@tempdima}{\pgf@yb}}
\pgfpathlineto{\pgfpoint{\pgf@xb + 5pt + \@tempdimb}{\pgf@yb}}
\pgfpathlineto{\pgfpoint{\pgf@xb + 5pt - \@tempdimb}{\pgf@ya}}
\pgfpathlineto{\pgfpoint{\pgf@xa - 5pt + \@tempdima}{\pgf@ya}}
\pgfpathclose
}
}}

\boxshape{NEbox}{0pt}{5pt}
\boxshape{SEbox}{0pt}{-5pt}
\boxshape{NWbox}{5pt}{0pt}
\boxshape{SWbox}{-5pt}{0pt}
\boxshape{EBox}{-3pt}{3pt}
\boxshape{WBox}{3pt}{-3pt}
\makeatother

\tikzstyle{cloud}=[shape=cloud,draw,minimum width=1.5cm,minimum height=1.5cm]

\tikzstyle{map}=[draw,shape=NEbox,inner sep=2pt,minimum height=6mm,fill=white]
\tikzstyle{dashedmap}=[draw,dashed,shape=NEbox,inner sep=2pt,minimum height=6mm,fill=white]
\tikzstyle{mapdag}=[draw,shape=SEbox,inner sep=2pt,minimum height=6mm,fill=white]
\tikzstyle{mapadj}=[draw,shape=SEbox,inner sep=2pt,minimum height=6mm,fill=white]
\tikzstyle{maptrans}=[draw,shape=SWbox,inner sep=2pt,minimum height=6mm,fill=white]
\tikzstyle{mapconj}=[draw,shape=NWbox,inner sep=2pt,minimum height=6mm,fill=white]

\tikzstyle{medium map}=[draw,shape=NEbox,inner sep=2pt,minimum height=6mm,fill=white,minimum width=7mm]
\tikzstyle{medium map dag}=[draw,shape=SEbox,inner sep=2pt,minimum height=6mm,fill=white,minimum width=7mm]
\tikzstyle{medium map adj}=[draw,shape=SEbox,inner sep=2pt,minimum height=6mm,fill=white,minimum width=7mm]
\tikzstyle{medium map trans}=[draw,shape=SWbox,inner sep=2pt,minimum height=6mm,fill=white,minimum width=7mm]
\tikzstyle{medium map conj}=[draw,shape=NWbox,inner sep=2pt,minimum height=6mm,fill=white,minimum width=7mm]
\tikzstyle{semilarge map}=[draw,shape=NEbox,inner sep=2pt,minimum height=6mm,fill=white,minimum width=9.5mm]
\tikzstyle{semilarge map trans}=[draw,shape=SWbox,inner sep=2pt,minimum height=6mm,fill=white,minimum width=9.5mm]
\tikzstyle{semilarge map adj}=[draw,shape=SEbox,inner sep=2pt,minimum height=6mm,fill=white,minimum width=9.5mm]
\tikzstyle{semilarge map dag}=[draw,shape=SEbox,inner sep=2pt,minimum height=6mm,fill=white,minimum width=9.5mm]
\tikzstyle{semilarge map conj}=[draw,shape=NWbox,inner sep=2pt,minimum height=6mm,fill=white,minimum width=9.5mm]
\tikzstyle{large map}=[draw,shape=NEbox,inner sep=2pt,minimum height=6mm,fill=white,minimum width=12mm]
\tikzstyle{large map conj}=[draw,shape=NWbox,inner sep=2pt,minimum height=6mm,fill=white,minimum width=12mm]
\tikzstyle{very large map}=[draw,shape=NEbox,inner sep=2pt,minimum height=6mm,fill=white,minimum width=17mm]

\tikzstyle{medium dmap}=[draw,doubled,shape=NEbox,inner sep=2pt,minimum height=6mm,fill=white,minimum width=7mm]
\tikzstyle{medium dmap dag}=[draw,doubled,shape=SEbox,inner sep=2pt,minimum height=6mm,fill=white,minimum width=7mm]
\tikzstyle{medium dmap adj}=[draw,doubled,shape=SEbox,inner sep=2pt,minimum height=6mm,fill=white,minimum width=7mm]
\tikzstyle{medium dmap trans}=[draw,doubled,shape=SWbox,inner sep=2pt,minimum height=6mm,fill=white,minimum width=7mm]
\tikzstyle{medium dmap conj}=[draw,doubled,shape=NWbox,inner sep=2pt,minimum height=6mm,fill=white,minimum width=7mm]
\tikzstyle{semilarge dmap}=[draw,doubled,shape=NEbox,inner sep=2pt,minimum height=6mm,fill=white,minimum width=9.5mm]
\tikzstyle{semilarge dmap trans}=[draw,doubled,shape=SWbox,inner sep=2pt,minimum height=6mm,fill=white,minimum width=9.5mm]
\tikzstyle{semilarge dmap adj}=[draw,doubled,shape=SEbox,inner sep=2pt,minimum height=6mm,fill=white,minimum width=9.5mm]
\tikzstyle{semilarge dmap dag}=[draw,doubled,shape=SEbox,inner sep=2pt,minimum height=6mm,fill=white,minimum width=9.5mm]
\tikzstyle{semilarge dmap conj}=[draw,doubled,shape=NWbox,inner sep=2pt,minimum height=6mm,fill=white,minimum width=9.5mm]
\tikzstyle{large dmap}=[draw,doubled,shape=NEbox,inner sep=2pt,minimum height=6mm,fill=white,minimum width=12mm]
\tikzstyle{large dmap conj}=[draw,doubled,shape=NWbox,inner sep=2pt,minimum height=6mm,fill=white,minimum width=12mm]
\tikzstyle{large dmap trans}=[draw,doubled,shape=SWbox,inner sep=2pt,minimum height=6mm,fill=white,minimum width=12mm]
\tikzstyle{large dmap adj}=[draw,doubled,shape=SEbox,inner sep=2pt,minimum height=6mm,fill=white,minimum width=12mm]
\tikzstyle{large dmap dag}=[draw,doubled,shape=SEbox,inner sep=2pt,minimum height=6mm,fill=white,minimum width=12mm]
\tikzstyle{very large dmap}=[draw,doubled,shape=NEbox,inner sep=2pt,minimum height=6mm,fill=white,minimum width=19.5mm]

\tikzstyle{muxbox}=[draw,shape=rectangle,minimum height=3mm,minimum width=3mm,fill=white]
\tikzstyle{dmuxbox}=[muxbox,doubled]

\tikzstyle{box}=[draw,shape=rectangle,inner sep=2pt,minimum height=6mm,minimum width=6mm,fill=white]
\tikzstyle{dbox}=[draw,doubled,shape=rectangle,inner sep=2pt,minimum height=6mm,minimum width=6mm,inline text,fill=white]
\tikzstyle{dmap}=[draw,doubled,shape=NEbox,inner sep=2pt,minimum height=6mm,fill=white]
\tikzstyle{dmapdag}=[draw,doubled,shape=SEbox,inner sep=2pt,minimum height=6mm,fill=white]
\tikzstyle{dmapadj}=[draw,doubled,shape=SEbox,inner sep=2pt,minimum height=6mm,fill=white]
\tikzstyle{dmaptrans}=[draw,doubled,shape=SWbox,inner sep=2pt,minimum height=6mm,fill=white]
\tikzstyle{dmapconj}=[draw,doubled,shape=NWbox,inner sep=2pt,minimum height=6mm,fill=white]

\tikzstyle{ddmap}=[draw,doubled,dashed,shape=NEbox,inner sep=2pt,minimum height=6mm,fill=white]
\tikzstyle{ddmapdag}=[draw,doubled,dashed,shape=SEbox,inner sep=2pt,minimum height=6mm,fill=white]
\tikzstyle{ddmapadj}=[draw,doubled,dashed,shape=SEbox,inner sep=2pt,minimum height=6mm,fill=white]
\tikzstyle{ddmaptrans}=[draw,doubled,dashed,shape=SWbox,inner sep=2pt,minimum height=6mm,fill=white]
\tikzstyle{ddmapconj}=[draw,doubled,dashed,shape=NWbox,inner sep=2pt,minimum height=6mm,fill=white]

\boxshape{sNEbox}{0pt}{3pt}
\boxshape{sSEbox}{0pt}{-3pt}
\boxshape{sNWbox}{3pt}{0pt}
\boxshape{sSWbox}{-3pt}{0pt}
\tikzstyle{smap}=[draw,shape=sNEbox,fill=white]
\tikzstyle{smapdag}=[draw,shape=sSEbox,fill=white]
\tikzstyle{smapadj}=[draw,shape=sSEbox,fill=white]
\tikzstyle{smaptrans}=[draw,shape=sSWbox,fill=white]
\tikzstyle{smapconj}=[draw,shape=sNWbox,fill=white]

\tikzstyle{dsmap}=[draw,dashed,shape=sNEbox,fill=white]
\tikzstyle{dsmapdag}=[draw,dashed,shape=sSEbox,fill=white]
\tikzstyle{dsmaptrans}=[draw,dashed,shape=sSWbox,fill=white]
\tikzstyle{dsmapconj}=[draw,dashed,shape=sNWbox,fill=white]

\boxshape{mNEbox}{0pt}{10pt}
\boxshape{mSEbox}{0pt}{-10pt}
\boxshape{mNWbox}{10pt}{0pt}
\boxshape{mSWbox}{-10pt}{0pt}
\tikzstyle{mmap}=[draw,shape=mNEbox]
\tikzstyle{mmapdag}=[draw,shape=mSEbox]
\tikzstyle{mmaptrans}=[draw,shape=mSWbox]
\tikzstyle{mmapconj}=[draw,shape=mNWbox]

\tikzstyle{mmapgray}=[draw,fill=gray!40!white,shape=mNEbox]
\tikzstyle{smapgray}=[draw,fill=gray!40!white,shape=sNEbox]

\makeatletter

\pgfdeclareshape{cornerpoint}{
\inheritsavedanchors[from=rectangle] % this is nearly a rectangle
\inheritanchorborder[from=rectangle]
\inheritanchor[from=rectangle]{center}
\inheritanchor[from=rectangle]{north}
\inheritanchor[from=rectangle]{south}
\inheritanchor[from=rectangle]{west}
\inheritanchor[from=rectangle]{east}
\backgroundpath{% this is new
\southwest \pgf@xa=\pgf@x \pgf@ya=\pgf@y
\northeast \pgf@xb=\pgf@x \pgf@yb=\pgf@y

\pgfmathsetmacro{\pgf@shorten@left}{\pgfkeysvalueof{/tikz/shorten left}}
\pgfmathsetmacro{\pgf@shorten@right}{\pgfkeysvalueof{/tikz/shorten right}}

\pgfpathmoveto{\pgfpoint{0.5 * (\pgf@xa + \pgf@xb)}{\pgf@ya - 5pt}}
\pgfpathlineto{\pgfpoint{\pgf@xa - 8pt + \pgf@shorten@left}{\pgf@yb - 1.5 * \pgf@shorten@left}}
\pgfpathlineto{\pgfpoint{\pgf@xa - 8pt + \pgf@shorten@left}{\pgf@yb}}
\pgfpathlineto{\pgfpoint{\pgf@xb + 8pt - \pgf@shorten@right}{\pgf@yb}}
\pgfpathlineto{\pgfpoint{\pgf@xb + 8pt - \pgf@shorten@right}{\pgf@yb - 1.5 * \pgf@shorten@right}}
\pgfpathclose
}
}

\pgfdeclareshape{cornercopoint}{
\inheritsavedanchors[from=rectangle] % this is nearly a rectangle
\inheritanchorborder[from=rectangle]
\inheritanchor[from=rectangle]{center}
\inheritanchor[from=rectangle]{north}
\inheritanchor[from=rectangle]{south}
\inheritanchor[from=rectangle]{west}
\inheritanchor[from=rectangle]{east}
\backgroundpath{% this is new
\southwest \pgf@xa=\pgf@x \pgf@ya=\pgf@y
\northeast \pgf@xb=\pgf@x \pgf@yb=\pgf@y

\pgfmathsetmacro{\pgf@shorten@left}{\pgfkeysvalueof{/tikz/shorten left}}
\pgfmathsetmacro{\pgf@shorten@right}{\pgfkeysvalueof{/tikz/shorten right}}

\pgfpathmoveto{\pgfpoint{0.5 * (\pgf@xa + \pgf@xb)}{\pgf@yb + 5pt}}
\pgfpathlineto{\pgfpoint{\pgf@xa - 8pt + \pgf@shorten@left}{\pgf@ya + 1.5 * \pgf@shorten@left}}
\pgfpathlineto{\pgfpoint{\pgf@xa - 8pt + \pgf@shorten@left}{\pgf@ya}}
\pgfpathlineto{\pgfpoint{\pgf@xb + 8pt - \pgf@shorten@right}{\pgf@ya}}
\pgfpathlineto{\pgfpoint{\pgf@xb + 8pt - \pgf@shorten@right}{\pgf@ya + 1.5 * \pgf@shorten@right}}
\pgfpathclose
}
}

\makeatother

\pgfkeyssetvalue{/tikz/shorten left}{0pt}
\pgfkeyssetvalue{/tikz/shorten right}{0pt}

\tikzstyle{kpoint common}=[draw,fill=white,inner sep=1pt,minimum height=4mm]
\tikzstyle{kpoint sc}=[shape=cornerpoint,kpoint common]
\tikzstyle{kpoint adjoint sc}=[shape=cornercopoint,kpoint common]
\tikzstyle{kpoint}=[shape=cornerpoint,shorten left=5pt,kpoint common]
\tikzstyle{kpoint adjoint}=[shape=cornercopoint,shorten left=5pt,kpoint common]
\tikzstyle{kpoint conjugate}=[shape=cornerpoint,shorten right=5pt,kpoint common]
\tikzstyle{kpoint transpose}=[shape=cornercopoint,shorten right=5pt,kpoint common]
\tikzstyle{kpoint symm}=[shape=cornerpoint,shorten left=5pt,shorten right=5pt,kpoint common]

\tikzstyle{black kpoint}=[shape=cornerpoint,shorten left=5pt,kpoint common,fill=black,font=\color{white}]
\tikzstyle{black kpoint adjoint}=[shape=cornercopoint,shorten left=5pt,kpoint common,fill=black,font=\color{white}]
\tikzstyle{black kpointadj}=[shape=cornercopoint,shorten left=5pt,kpoint common,fill=black,font=\color{white}]

\tikzstyle{black dkpoint}=[shape=cornerpoint,shorten left=5pt,kpoint common,fill=black, doubled,font=\color{white}]
\tikzstyle{black dkpoint adjoint}=[shape=cornercopoint,shorten left=5pt,kpoint common,fill=black, doubled,font=\color{white}]
\tikzstyle{black dkpointadj}=[shape=cornercopoint,shorten left=5pt,kpoint common,fill=black, doubled,font=\color{white}] 

\tikzstyle{kpointdag}=[kpoint adjoint]
\tikzstyle{kpointadj}=[kpoint adjoint]
\tikzstyle{kpointconj}=[kpoint conjugate]
\tikzstyle{kpointtrans}=[kpoint transpose]

\tikzstyle{big kpoint}=[kpoint, minimum width=1.2 cm, minimum height=8mm, inner sep=4pt, text depth=3mm]

\tikzstyle{wide kpoint}=[kpoint, minimum width=1 cm, inner sep=2pt]%, text depth=-0.7 mm]
\tikzstyle{wide kpointdag}=[kpointdag, minimum width=1 cm, inner sep=2pt]%, text depth=0.7 mm]
\tikzstyle{wide kpointconj}=[kpointconj, minimum width=1 cm, inner sep=2pt]%, text depth=-0.7 mm]
\tikzstyle{wide kpointtrans}=[kpointtrans, minimum width=1 cm, inner sep=2pt]%, text depth=0.7 mm]

\tikzstyle{gray kpoint}=[kpoint,fill=gray!50!white]
\tikzstyle{gray kpointdag}=[kpointdag,fill=gray!50!white]
\tikzstyle{gray kpointadj}=[kpointadj,fill=gray!50!white]
\tikzstyle{gray kpointconj}=[kpointconj,fill=gray!50!white]
\tikzstyle{gray kpointtrans}=[kpointtrans,fill=gray!50!white]

\tikzstyle{gray dkpoint}=[kpoint,fill=gray!50!white,doubled]
\tikzstyle{gray dkpointdag}=[kpointdag,fill=gray!50!white,doubled]
\tikzstyle{gray dkpointadj}=[kpointadj,fill=gray!50!white,doubled]
\tikzstyle{gray dkpointconj}=[kpointconj,fill=gray!50!white,doubled]
\tikzstyle{gray dkpointtrans}=[kpointtrans,fill=gray!50!white,doubled]

\tikzstyle{white label}=[draw,fill=white,rectangle,inner sep=0.7 mm]
\tikzstyle{gray label}=[draw,fill=gray!50!white,rectangle,inner sep=0.7 mm]
\tikzstyle{black label}=[draw,fill=black,rectangle,inner sep=0.7 mm]

\tikzstyle{dkpoint}=[kpoint,doubled]
\tikzstyle{wide dkpoint}=[wide kpoint,doubled]
\tikzstyle{dkpointdag}=[kpoint adjoint,doubled]
\tikzstyle{wide dkpointdag}=[wide kpointdag,doubled]
\tikzstyle{dkcopoint}=[kpoint adjoint,doubled]
\tikzstyle{dkpointadj}=[kpoint adjoint,doubled]
\tikzstyle{dkpointconj}=[kpoint conjugate,doubled]
\tikzstyle{dkpointtrans}=[kpoint transpose,doubled]

\tikzstyle{kscalar}=[kpoint common, shape=EBox, inner xsep=-1pt, inner ysep=3pt,font=\small]
\tikzstyle{kscalarconj}=[kpoint common, shape=WBox, inner xsep=-1pt, inner ysep=3pt,font=\small]

\tikzstyle{spekpoint}=[kpoint sc,minimum height=5mm,inner sep=3pt]
\tikzstyle{spekcopoint}=[kpoint adjoint sc,minimum height=5mm,inner sep=3pt]

\tikzstyle{dspekpoint}=[spekpoint,doubled]
\tikzstyle{dspekcopoint}=[spekcopoint,doubled]

 \tikzstyle{upground}=[circuit ee IEC,thick,ground,rotate=90,scale=1.5]
 \tikzstyle{downground}=[circuit ee IEC,thick,ground,rotate=-90,scale=1.5]
 \tikzstyle{bigground}=[regular polygon,regular polygon sides=3,draw=gray,scale=0.50,inner sep=-0.5pt,minimum width=10mm,fill=gray]
\tikzstyle{arrs}=[-latex,font=\small,auto]
\tikzstyle{arrow plain}=[arrs]
\tikzstyle{arrow dashed}=[dashed,arrs]
\tikzstyle{arrow bold}=[very thick,arrs]
\tikzstyle{arrow hide}=[draw=white!0,-]
\tikzstyle{arrow reverse}=[latex-]
\tikzstyle{cdnode}=[]

\newcommand{\dotonly}[1]{%
\,\begin{tikzpicture}[dotpic]
\node [#1] (a) at (0,0) {};
\end{tikzpicture}\,}
\newcommand{\whitedot}{\dotonly{white dot}\xspace}
\newcommand{\graydot}{\dotonly{gray dot}\xspace}
\let\olddagger\dagger
\renewcommand{\dagger}{\ensuremath{\olddagger}\xspace}

\theoremstyle{definition}
\newtheorem{theorem}{Theorem}[section]
\newtheorem*{theorem*}{Theorem}

\newtheorem{example*}[theorem]{Example*}
\newtheorem{examples*}[theorem]{Examples*}
\newtheorem{remark}[theorem]{Remark}
\newtheorem{remark*}[theorem]{Remark*}

\usepackage[color]{changebar}

\usepackage{color}
\def\bR{\begin{color}{red}} 
\def\bB{\begin{color}{blue}}
\def\bM{\begin{color}{magenta}}
\def\bC{\begin{color}{cyan}}
\def\bW{\begin{color}{white}}
\def\bBl{\begin{color}{black}} 
\def\bG{\begin{color}{green}}
\def\bY{\begin{color}{yellow}}
\def\e{\end{color}\xspace}
\newcommand{\bit}{\begin{itemize}}
\newcommand{\eit}{\end{itemize}\par\noindent}
\newcommand{\ben}{\begin{enumerate}}
\newcommand{\een}{\end{enumerate}\par\noindent}
\newcommand{\beq}{\begin{equation}}
\newcommand{\eeq}{\end{equation}\par\noindent}
\newcommand{\beqa}{\begin{eqnarray*}}
\newcommand{\eeqa}{\end{eqnarray*}\par\noindent}
\newcommand{\beqn}{\begin{eqnarray}}
\newcommand{\eeqn}{\end{eqnarray}\par\noindent}
\usepackage[numbers,sort&compress]{natbib}

\title{Picture-perfect Quantum Key Distribution}
\author{Aleks Kissinger}
\affiliation{Institute for Computing and Information Sciences, Radboud University Nijmegen}
\email{aleks@cs.ru.nl}

\author{Sean Tull}
\affiliation{Department of Computer Science, University of Oxford}
\email{sean.tull@cs.ox.ac.uk}
\author{Bas Westerbaan}
\affiliation{Institute for Computing and Information Sciences, Radboud University Nijmegen}
\email{bas@westerbaan.name}
\date{}

\begin{document}

\maketitle

\begin{abstract}
We provide a new way to bound the security of quantum key distribution using only two high-level, diagrammatic features of quantum processes: the compositional behavior of complementary measurements and the essential uniqueness of purification. We begin by demonstrating a proof in the simplest case, where the eavesdropper doesn't noticeably disturb the channel at all and has no quantum memory. We then show how this approach extends straightforwardly to account for an eavesdropper with quantum memory and the presence of noise.
\end{abstract}

\section{Introduction}

Traditionally in cryptography, parties use pre-shared keys to communicate securely. In 1976 Diffie and Hellman introduced the first \emph{key agreement protocol} that allows two parties, say Alice and Bob, to securely establish a key over an insecure channel~\cite{nd}. There are two caveats to its security. First, it's an unauthenticated key agreement protocol,\footnote{%
   An attacker Eve that can intercept and modify all communication
   between Alice and Bob, can simply impersonate Bob and perform all
   the steps Bob would.  Now Alice thinks she established a shared key
   with Bob, where in reality she established a shared key with Eve.}
% Secure telephone systems often use a key agreement protocol
%    and solve this authentication problem
%    by displaying a code derived from the established key on the display:
% if Eve is in the middle,
%    then Alice and Bob will have a different security code on their screen
%    as the keys they separately established with Eve are different.
and second, its security hinges on the difficulty of computing discrete logarithms. Famously, Shor showed in 1994 \cite{Shor} that a (large-scale stable) quantum computer is able to calculate discrete logarithms with ease, breaking Diffie and Hellman's original scheme.

There are other key agreement protocols whose
    security is based on mathematical problems
    which are believed to  be difficult,
    even for quantum computers.
    An example is \textsc{NewHope}~\cite{newhope,googlenh}.
    Such schemes form part of the program of \emph{post-quantum crypto\-graphy}~\cite{pqcrypto}.

On the other hand Bennett and Brassard
    proposed a completely different unauthenticated key agreement
    protocol, now called BB84~\cite{BB84}, 
whose security is not based on a mathematical problem, but on physical assumptions on quantum mechanical systems. In this protocol Alice needs to be able to send Bob qubits. In most implementations, like~\cite{idq}, this is done by sending photons over a fiber optic cable, where the qubits are encoded in the polarization of the photons. After BB84 other variations have been proposed, notably E91~\cite{Ekert91}, and they are all indiscriminately referred to as Quantum Key Distribution (QKD).

In their original paper, Bennett and Brassard give a short argument for the security of BB84. Feeling this proof was not satisfactory, subsequent authors have proposed more meticulous proofs, for instance~\cite{mayers}. However, this added rigor came at the cost of complexity, prompting Shor and Preskill to publish a `simple proof'~\cite{simpleproof}.
This did not settle the matter: many subsequent publications have appeared on the security of BB84 and its variants, differing not only in levels of rigor vs.~simplicity, but also in the physical assumptions and the efficiency of the protocol itself (e.g.~qubits required per bit of shared key)~\cite{renner,qkdsec1,qkdsec2,qkdsec3,qkdsec4,qkdsec5,qkdsec6,qkdsec7,qkdsec8,qkdsec9}.
%A noteworthy entry among these is Renner's Ph.D-thesis.

The present paper contributes another proof to the mix, whose primary aim is simplicity. It is quite different from those that came before, in that we adopt a purely graphical notation and rely on techniques which arose in categorical quantum mechanics~\cite{AC1} and the diagrammatic approach to quantum theory~\cite{CKbook}. We will show that two high-level features of quantum processes--namely the behavior of complementary measurements under composition~\cite{CD2} and the essential uniqueness of purification of quantum channels~\cite{QuantumFromPrinciples}--suffice to show that there exists no quantum process with which an eavesdropper can undetectably extract information about Alice and Bob's shared key.

%The present paper aims for simplicity. Our tool is the graphical notation for quantum mechanical processes.

We begin by briefly reviewing the graphical notation for quantum processes, including depictions of purification and complementary measurements. In Section~\ref{sec:qkd}, we describe a version of BB84 in this notation and give a one-line version of a proof of correctness which already appears in the literature~\cite{BWWWZ,CKbook}, but makes the (unreasonably) strong assumption that Eve is only allowed to measure and re-prepare in the $Z$ and $X$ bases. Our first main result removes this assumption, allowing Eve to perform any quantum process to (attempt to) extract information from Alice and Bob's channel:
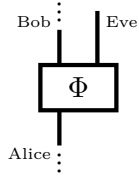
\ctikzfig{eve}
We then formulate the condition that Eve's channel remains undetected by means of two equations, corresponding to the cases where Alice and Bob's measurement choices agree:
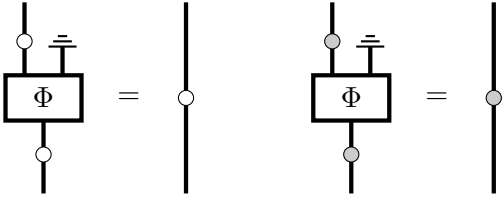
\[
\input{exact-req1.tikz}\qquad \qquad
\input{exact-req2.tikz}
\]
Using the behavior of complementary measurements, we show that this implies that Eve's channel separates:
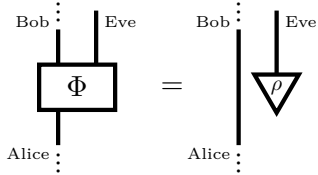
\ctikzfig{eve-sep}

While already much more general than previous graphical proofs, this still makes two very strong assumptions. First, it assumes that Eve performs the same process each time Alice and Bob use their channel. Second, it assumes all equalities are exact, so it offers no guarantees in the presence of noise.

We remove the first limitation in Section~\ref{sec:memory} by allowing Eve to maintain an arbitrarily large quantum memory between usages of the channel, and show that the separation argument still holds. We remove the second limitation in Section~\ref{sec:noise} by showing the same graphical proof from Section~\ref{sec:qkd} goes through, thanks to the continuity properties of purification and diagram rewriting, if we replace equality by $\epsilon$-closeness in the completely bounded norm (written `$\,\cdots \underset{\mathrm{cb}}{\overset{\varepsilon}{=}} \cdots\,$'). This enables us to give a security bound for the QKD protocol comparable to the one suggested in~\cite{contstinespring}. In particular, the fact that Eve's process disturbs the channel very little can be formulated as:
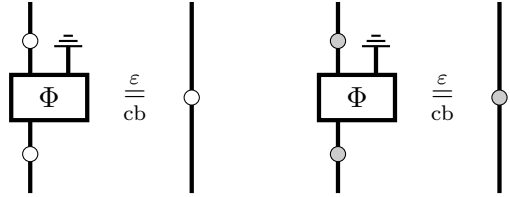
\[
\input{approx-req1.tikz}\qquad\qquad
\input{approx-req2.tikz}
\]
This implies that Eve's process is within $N\sqrt{\epsilon}$ of a separable one:
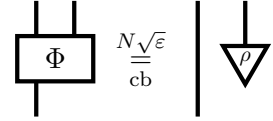
\[
\input{separates-approx.tikz}
\]
for some constant $N$ depending only on the dimension of Alice and Bob's system.

%We organized the paper as follows. First we introduce the graphical notation.
%Then we describe a variant of BB84 using it and give the proof for the simplest case. In consequent sections we weaken the assumptions by first allowing the attacker arbitrary (possibly weak) measurement; then a memory and finally noise.

\section{Preliminaries}\label{sec:prelims}

Throughout the paper, we will use \textit{string diagram} notation
for linear maps and channels. \cite{CKbook} Systems are depicted as wires and maps as boxes. To make it clear whether we are working with linear maps between Hilbert spaces or completely positive maps (CP-maps),
we will depict finite-dimensional Hilbert spaces $H, K, \ldots$ as thin wires and the associated spaces of operators $\mathcal B(H), \mathcal B(K), \ldots$ as thick wires.

A linear map $V \colon H \to K$ and CP-map $\Phi \colon \mathcal B(H) \to \mathcal B(K)$ (for Hilbert spaces $H,K$) are depicted respectively as:
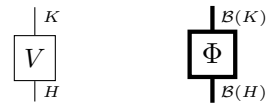
\[
\input{lin-map.tikz}
\qquad\qquad
\input{cp-map.tikz}
\]
We omit wire labels if they are irrelevant or clear from context.
Compositions of maps is depicted by plugging boxes together vertically:
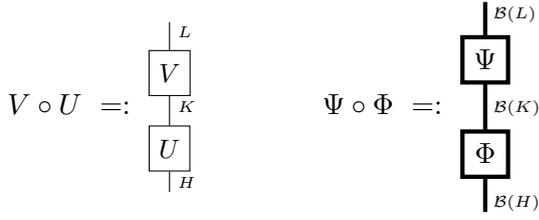
\[
V \circ U \ =:\ \input{lin-map-comp.tikz}
\qquad\qquad
\Psi \circ \Phi \ =:\ \input{cp-map-comp.tikz}
\]
and tensor products by putting boxes side-by-side:
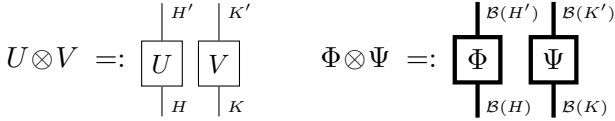
\[
U \otimes V \ =:\ \input{lin-map-tensor.tikz}
\qquad
\Phi \otimes \Psi \ =:\ \input{cp-map-tensor.tikz}
\]
Similarly, maps between tensor products such as $U \colon H \to K \otimes L$ or $\Phi \colon \mathcal B(H ) \to \mathcal B(K) \otimes \mathcal B (L)$ are depicted as boxes with multiple input/output wires:
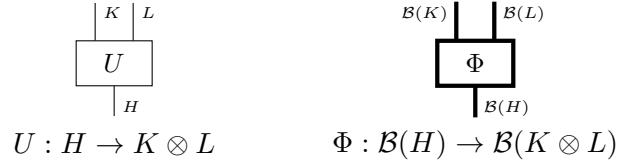
\[
\input{lin-map-multi.tikz}
\qquad\qquad
\input{cp-map-multi.tikz}
\]
where we identify $\mathcal B(H) \otimes \mathcal B(K) \cong \mathcal B(H \otimes K)$. Identity linear maps/CP-maps are represented as `plain wires':
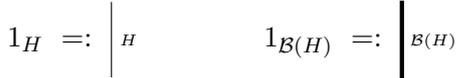
\[
1_H\ =:\ \begin{tikzpicture}
	\begin{pgfonlayer}{nodelayer}
		\node [style=none] (0) at (0, 1) {};
		\node [style=none] (1) at (0, -1) {};
		\node [style=right label] (2) at (0.25, 0) {$H$};
	\end{pgfonlayer}
	\begin{pgfonlayer}{edgelayer}
		\draw (1.center) to (0.center);
	\end{pgfonlayer}
\end{tikzpicture}
\qquad\qquad
1_{\mathcal B(H)} \ =:\ \begin{tikzpicture}
	\begin{pgfonlayer}{nodelayer}
		\node [style=none] (0) at (0, 1) {};
		\node [style=none] (1) at (0, -1) {};
		\node [style=right label] (2) at (0.25, 0) {$\mathcal B(H)$};
	\end{pgfonlayer}
	\begin{pgfonlayer}{edgelayer}
		\draw [style=boldedge] (1.center) to (0.center);
	\end{pgfonlayer}
\end{tikzpicture}
\]
since $U \circ 1_H = 1_K \circ U = U$ and similarly for CP-maps.
The trivial system $\mathbb C$ is depicted as `no wire', since $H \otimes \mathbb C \cong H$ and $\mathcal B(H) \otimes \mathbb C \cong \mathcal B(H)$. Regarding vectors as linear maps $v \colon \mathbb C \to H$ and positive operators as CP-maps $\rho : \mathbb C \to \mathcal B(H)$, we can depict vectors $\ket{\psi} \in H$ and quantum states in $\rho \in \mathcal B(H)$ respectively as maps from 0 to 1 wire:
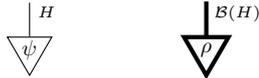
\[
\begin{tikzpicture}
	\begin{pgfonlayer}{nodelayer}
		\node [style=point] (0) at (0, -0.25) {$\psi$};
		\node [style=right label] (1) at (0.25, 0.75) {$H$};
		\node [style=none] (2) at (0, 1) {};
	\end{pgfonlayer}
	\begin{pgfonlayer}{edgelayer}
		\draw (0) to (2.center);
	\end{pgfonlayer}
\end{tikzpicture}
\qquad\qquad
\begin{tikzpicture}
	\begin{pgfonlayer}{nodelayer}
		\node [style=dpoint] (0) at (0, -0.25) {$\rho$};
		\node [style=right label] (1) at (0.25, 0.75) {$\mathcal B(H)$};
		\node [style=none] (2) at (0, 1) {};
	\end{pgfonlayer}
	\begin{pgfonlayer}{edgelayer}
		\draw [style=boldedge] (0) to (2.center);
	\end{pgfonlayer}
\end{tikzpicture}
\]
Similarly, we can depict linear functionals $\bra{\psi} \colon H \to \mathbb C$ and CP-maps of the form $\bm e \colon \mathcal B(H) \to \mathbb C$ as maps from 1 to 0 wires:
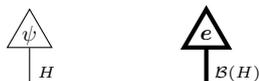
\[
\begin{tikzpicture}
	\begin{pgfonlayer}{nodelayer}
		\node [style=copoint] (0) at (0, 0.5) {$\psi$};
		\node [style=right label] (1) at (0.25, -0.5) {$H$};
		\node [style=none] (2) at (0, -0.75) {};
	\end{pgfonlayer}
	\begin{pgfonlayer}{edgelayer}
		\draw (0) to (2.center);
	\end{pgfonlayer}
\end{tikzpicture}
\qquad\qquad
\begin{tikzpicture}
	\begin{pgfonlayer}{nodelayer}
		\node [style=dcopoint] (0) at (0, 0.5) {$\bm e$};
		\node [style=right label] (1) at (0.25, -0.5) {$\mathcal B(H)$};
		\node [style=none] (2) at (0, -0.75) {};
	\end{pgfonlayer}
	\begin{pgfonlayer}{edgelayer}
		\draw [style=boldedge] (0) to (2.center);
	\end{pgfonlayer}
\end{tikzpicture}
\]
%Sean: Probably don't need this actually.
% As special cases, we may view complex numbers $z \in \mathbb{C}$ as linear maps $\mathbb{C} \to \mathbb{C}$ via $w \mapsto z\cdot w$, and positive reals $r$ as CP-maps of this kind, depicting each as maps from 0 to 0 wires: 
% \[
% \tikzfig{scalars}
% \]

\subsection{Purification}

The linear map that sends $\rho \in \mathcal B(H)$ to its trace is a CP-map $\Tr \colon \mathcal B(H) \to \mathbb C$, which we draw as a `ground' symbol:
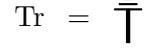
\[
%1_H\ =\ \maxmix \qquad\qquad
\Tr\ =\ \discard
\]
%Todo: might need to mention that Tr_{A \otimes B} = Tr{A} \otimes Tr{B}

Using this notation, we can express the property of a CP-map being trace-preserving as follows:
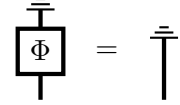
\ctikzfig{causal}

Furthermore, any linear map $U \colon H \to K$ induces a CP-map $\widehat U \colon \mathcal B(H) \to \mathcal B(K)$ via:
\begin{equation}\label{eq:pure}
\widehat U(\rho) = U \,\rho \,U^{\dagger}
\end{equation}
We call a CP-map \textit{pure} if it is of the form of \eqref{eq:pure}. We call this method of turning a linear map into a pure CP-map \textit{doubling}.

An important theorem about CP-maps is that any CP-map can be represented in an essentially unique manner, by means of a pure CP-map and the trace.

\begin{theorem}[Stinespring / purification, \cite{Stinespring,Chiri1}]
For any completely positive map $\Phi \colon \mathcal B (H)\to \mathcal B (K)$ there is a Hilbert space~$L$ and linear map $V \colon H \to K \otimes L$ with 
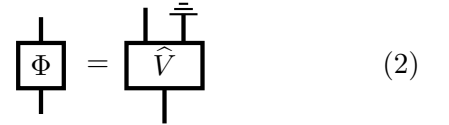
\begin{equation}\label{eq:purification}
\input{purification.tikz}
\end{equation}
Moreover, for any (other) linear map $V' \colon H \to K \otimes L$ satisfying \eqref{eq:purification}, there is a unitary $U \colon L \to L$ such that:
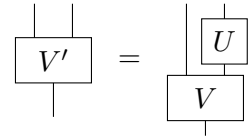
\ctikzfig{unique-purification}
\end{theorem}

\subsection{Spiders and decoherence}

To each orthonormal basis (ONB) $\{ \ket i \} $ of  a (finite-dimensional) Hilbert space~$H$, we associate a family of linear maps called \textit{spiders} as follows:
%Todo Sean: we probably don't need spiders in full, just the multiplication + unit + comultiplication + counit, plus the unit laws. The main proof doesn't use commutativity or even associativity of spiders.
% We mainly only use spiders in their `pure' form. 

\[ \whitedot_m^n :=\ \input{spider-mn.tikz}\ =\ \sum_i \underbrace{\ket{i\ldots i}}_n\underbrace{\bra{i\ldots i}}_m \]
where a spider with zero legs is a complex number $D$, the dimension of the Hilbert space. 

Different ONBs induce different spiders, and furthermore an ONB is uniquely fixed by its family of spiders \cite{CPV}.
%Clearly $(S_m^n)^\dagger = S_n^m$ and spiders are invariant under interchanging inputs and outputs. 
%The other relevant law is that 
Adjacent spiders associated with the same ONB fuse together. That is, for $k \geq 1$, we have:
\begin{equation}
\label{eq:spider1}
\input{spidercompcount.tikz}\ =\ 
\input{spidernnmm.tikz}
\end{equation}
We write the doubled spiders as follows:
\[
\whitedot_m^n := \input{spider.tikz} \quad\mapsto\quad
\widehat{\whitedot_m^n} := \input{qspider.tikz}
\]
These CP-maps, which satisfy the same fusion law, are called \textit{quantum spiders}. Examining the concrete expression:
\begin{align*}
\widehat{\whitedot_m^n}(\rho) \  :=\ & (\whitedot_m^n) \, \rho\,  (\whitedot_m^n)^\dagger \\ 
   \ =\ &\ 
   \sum_{ij}  \underbrace{\ket{i\ldots i}}_n\underbrace{\bra{i\ldots i}}_m\, \rho\, \underbrace{\ket{j\ldots j}}_m\underbrace{\bra{j\ldots j}}_n 
   \end{align*}
we see that:
\[ \Tr(\widehat{\whitedot_m^n}(\rho)) \ =\ 
\sum_{i} \underbrace{\bra{i\ldots i}}_m\, \rho\, \underbrace{\ket{i\ldots i}}_m
\]
Hence, such a map is trace-preserving whenever it has exactly one input:
\begin{equation}\label{eq:tp-one-input}
  \input{tp-one-input.tikz}
\end{equation}

One derived map which will be particularly important in the sequel is the \textit{decoherence map}, which arises from tracing out one output of $\widehat{\whitedot_1^2}$:
\[
\bm d_{\whitedot} \ =\ \input{decoh-def.tikz}
\]
Concretely:
\[
\bm d_{\whitedot}(\rho) =
\sum_i \ketbra{i}{i}\, \rho\, \ketbra{i}{i}
\]
That is, it projects a positive matrix written with respect to the ONB $\{ \ket i \} \subset H$ to its diagonal entries:
\[ \rho_{ij} \mapsto \delta_{ij} \rho_{ii} \]
The fact that this CP-map is trace-preserving follows from \eqref{eq:tp-one-input}:
\ctikzfig{decoh-causal}
and idempotence from spider-fusion and \eqref{eq:tp-one-input}:
\ctikzfig{decoh-idem}

If $\dim(H) = n$, decoherence gives a rank-$n$ projector in the $n^2$-dimensional space $\mathcal B(H)$, hence we can split this projector using the following linear maps:
\[
  \bm m_{\whitedot} \ =\  \input{meas-types.tikz} \qquad\qquad
  \bm e_{\whitedot} \ =\  \input{enc-types.tikz}
\]
where:
\[
\bm m_{\whitedot}(\ketbra{i}{j}) := \delta_{ij} \ket{i}
\qquad\qquad
\bm e_{\whitedot}(\ket{i}) := \ketbra{i}{i}
\]
Then we have:
\begin{equation}\label{eq:decoh-split}
  \input{decoh-split.tikz} \qquad\textrm{and}\qquad
\input{decoh-split2.tikz}
\end{equation}

Treating $\mathcal B(H)$ itself as a Hilbert space using the Hilbert-Schmidt inner product $\braket{U}{V} := \Tr(U^\dagger V)$, one can easily verify that $\bm m_{\whitedot} = (\bm e_{\whitedot})^\dagger$.

These maps have a straightforward operational interpretation which will play a role in the following sections. The first, $\bm m_{\whitedot}$, sends a quantum state to a classical probability distribution, whose entries are the Born rule probabilities associated with an ONB measurement of $\{\ket{i}\}$:
\[
\bm m_{\whitedot}(\rho) = \left(
\begin{matrix}
  \bra{1}\rho\ket{1} \\
  \vdots \\
  \bra{n}\rho\ket{n}
\end{matrix}
\right)
\]
Hence, we can think of it as a map from a quantum space $\mathcal B(H)$ to a classical space of probability distributions in $H$ (treated simply as a vector space), representing the act measuring in the given basis.
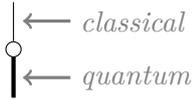
\ctikzfig{meas-cq}

Conversely, $\bm e_{\whitedot}$ sends a classical value (or distribution over classical values) to an encoding of that value as a quantum state:
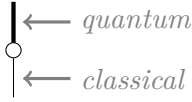
\ctikzfig{enc-cq}
Hence, it represents the act of \textit{encoding} a classical value with respect to an ONB of quantum states.

With these interpretations in mind, the left equation in \eqref{eq:decoh-split} gives an operational reading for decoherence. Namely, it arises from measuring a quantum system, followed by re-preparing it in the same basis.

Furthermore, certain spiders take on a special meaning as operations on classical systems. Namely,
\begin{equation}\label{eq:copy}
  \input{copy.tikz}\ =\ \sum_i \ketbra{ii}{i}
\end{equation}
is the process which \textit{copies} a classical value,
\begin{equation}\label{eq:delete}
  \begin{tikzpicture}
	\begin{pgfonlayer}{nodelayer}
		\node [style=none] (0) at (0, -0.75) {};
		\node [style=white dot] (1) at (0, 0.25) {};
	\end{pgfonlayer}
	\begin{pgfonlayer}{edgelayer}
		\draw (0.center) to (1);
	\end{pgfonlayer}
\end{tikzpicture}\ =\ \sum_i \bra{i}
\end{equation}
is \textit{deleting} (a.k.a. marginalisation), and
\begin{equation}\label{eq:uniform}
  \oneoverD\ \begin{tikzpicture}
	\begin{pgfonlayer}{nodelayer}
		\node [style=none] (0) at (0, 0.75) {};
		\node [style=white dot] (1) at (0, -0.25) {};
	\end{pgfonlayer}
	\begin{pgfonlayer}{edgelayer}
		\draw (0.center) to (1);
	\end{pgfonlayer}
\end{tikzpicture}\ =\ \oneoverD\ \sum_i \ket{i}
\end{equation}
is the \textit{uniform probability distribution}. Note that \eqref{eq:delete} and \eqref{eq:uniform} are the classical analogue to the trace and the maximally mixed state, respectively. Since \eqref{eq:copy} is a cloning (i.e. broadcasting) operation, it has no quantum analogue.

The final spider-derived map we will use is a \textit{non-demolition measurement}:
\ctikzfig{non-demo}
This map captures the process where we perform an ONB measurement (which produces classical data) but also leave the quantum system intact. Alternatively, we can read the RHS above literally as performing a (demolition) measurement, copying the measurement outcome, and using one of the copies to re-prepare the quantum system.

% Diagonal matrices afford a natural means of expressing classical states, i.e. probability distributions, as positive operators in $\mathcal B(\mathbb C^n)$:
% \[
% \left(
% \begin{matrix}
% P(1) \\
% \vdots \\
% P(n) \\
% \end{matrix}
% \right)
% \leftrightarrow
% \left(
% \begin{matrix}
% P(1)   & \cdots & 0 \\
% \vdots & \ddots & \vdots  \\
% 0      & \cdots & P(n) \\
% \end{matrix}
% \right)
% \]
% Traditionally, $\bm d_{\whitedot}$ represents the loss a quantum phase data with respect to a given ONB. However, if we consider it as a map from a \textit{quantum} state space and a \textit{classical} one:
% \ctikzfig{decoh-meas}
% this affords an additional interpretation. The map $\bm d_{\whitedot}$ sends a state $\rho$ to a probability distribution over its Born-rule probabilities with respect to the ONB $\{\ket{i}\}$:
% \[ \rho \mapsto \left(
% \begin{matrix}
% \bra{1}\rho\ket{1}   & \cdots & 0 \\
% \vdots & \ddots & \vdots  \\
% 0      & \cdots & \bra{n}\rho\ket{n} \\
% \end{matrix}
% \right) \]
% That is, $\bm d_{\whitedot}$ represents a quantum measurement.

% On the other hand, we can interpret the input system as classical on the output as quantum:
% \ctikzfig{decoh-enc}
% in which case $\bm d_{\whitedot}$ represents the process of \textit{encoding} a classical value as a quantum state, with respect to a given ONB of quantum states.

\subsection{Complementary spiders}

We can study the interaction of distinct ONBs on the same Hilbert space by introducing distinct families of spiders. From hence forth, we will fix two ONBs $\{ \ket{z_i} \}$, $\{ \ket{x_i} \}$ of a Hilbert space $H$ of dimension $D$, and let:
\begin{align*}
  \whitedot_m^n & :=\ \input{spider-mn.tikz}\ =\ \sum_i \underbrace{\ket{z_i\ldots z_i}}_n\underbrace{\bra{z_i\ldots z_i}}_m \\[4mm]
  \graydot_m^n & :=\ \input{gray-spider-mn.tikz}\ =\ \sum_i \underbrace{\ket{x_i\ldots x_i}}_n\underbrace{\bra{x_i\ldots x_i}}_m
\end{align*}
Consequently, we let:
\begin{align*}
\bm m_{\whitedot}(\ketbra{z_i}{z_j}) & := \delta_{ij} \ket{z_i}
&
\bm e_{\whitedot}(\ket{z_i}) & := \ketbra{z_i}{z_i} \\
\bm m_{\graydot}(\ketbra{x_i}{x_j}) & := \delta_{ij} \ket{x_i}
&
\bm e_{\graydot}(\ket{x_i}) & := \ketbra{x_i}{x_i}
\end{align*}

Two ONBs are said to be \emph{mutually unbiased} or \emph{complementary} whenever any member of one basis gives equal probabilities to all outcomes of a measurement with respect to the other, i.e.~we have $|\braket{z_i}{x_j}|^2 = \frac{1}{D}$, for all $i, j$. We can also express this property succinctly in terms of the measurement/encoding maps associated with a pair of spiders.

\begin{theorem}[\cite{CKbook}]\label{thm:complementary1}
  Two ONBs are mutually unbiased if and only if their measurement/encoding maps satisfy the following equation:
  \begin{equation}\label{eq:complementaritythick}
  \input{complementaritythick.tikz}
  \end{equation}
\end{theorem}

\begin{proof} Precomposing the LHS of \eqref{eq:complementaritythick} with $\ket{z_i}$ and post-composing with $\bra{x_j}$ yields:
  \[
  \bra{x_j} \bm m_{\graydot} \bm e_{\whitedot} \ket{z_i} =
  (\bm e_{\graydot} \ket{x_j})^\dagger (\bm e_{\whitedot} \ket{z_i}) \]
  \[
  = \Tr(\ketbra{x_j}{x_j}\ketbra{z_i}{z_i}) = |\braket{z_i}{x_j}|^2
  \]
  By definition of spiders, performing the same pre- and post-composition on the RHS yields \oneoverD, which completes the proof.
\end{proof}

%%%%%%%% <new>

Rather than resorting to measurement and encoding maps, we can also express mutual unbiasedness directly in terms of spiders with the help of an additional linear map, called the \textit{antipode} associated with a pair of spiders:
\[
\input{dualiser-def.tikz} \ = \ \sum_{ij} \braket{z_i}{x_j} \ketbra{x_j}{z_i}
\]

\begin{theorem}[\cite{CD2}] \label{thm:complementary2}
  Two ONBs are mutually unbiased if and only if their associated spiders satisfy the following equation:
  \begin{equation} \label{eq:complementarity}
  \input{complementarity.tikz}
\end{equation}
\end{theorem}

\begin{proof}
  Follows similarly to the proof of Theorem~\ref{thm:complementary1}. Pre-composing the LHS of \eqref{eq:complementarity} with $\ket{z_i}$ and post-composing with $\bra{x_j}$ yields:
  \[ \bra{x_j}s\ket{z_i} \braket{x_j}{z_i} = \braket{z_i}{x_j} \braket{x_j}{z_i} = |\braket{x_j}{z_i}|^2 \]
  Whereas pre- and post-composing on the RHS again yields $\oneoverD$.
\end{proof}

\begin{remark}
  The two conditions for unbiasedness given by Theorems \ref{thm:complementary1} and \ref{thm:complementary2} are related to each other via the (basis-dependent) isomorphism $\mathcal B(H) \cong H \otimes H$:
  \[ \sum_{ij} \rho_{ij} \ketbra{z_j}{z_i}
     \ \ \overset{\sim}{\longleftrightarrow} \ \ 
     \sum_{ij} \rho_{ij} \ket{z_i} \otimes \ket{z_j} \]
  Since the above isomorphism is defined with respect to the \whitedot-basis, the map $s$ corrects the \graydot-spider to account for this basis-dependence:
  \[
  \begin{tikzpicture}
	\begin{pgfonlayer}{nodelayer}
		\node [style=white dot] (0) at (0, 0) {};
		\node [style=none] (1) at (0, -1) {};
		\node [style=none] (2) at (0, 1) {};
	\end{pgfonlayer}
	\begin{pgfonlayer}{edgelayer}
		\draw [style=boldedge] (2.center) to (0);
		\draw (0) to (1.center);
	\end{pgfonlayer}
\end{tikzpicture}
  \ \overset{\sim}{\longleftrightarrow}\ 
  \input{copy.tikz}
  \qquad\qquad
  \begin{tikzpicture}
	\begin{pgfonlayer}{nodelayer}
		\node [style=gray dot] (0) at (0, 0) {};
		\node [style=none] (1) at (0, 1) {};
		\node [style=none] (2) at (0, -1) {};
	\end{pgfonlayer}
	\begin{pgfonlayer}{edgelayer}
		\draw [style=boldedge] (2.center) to (0);
		\draw (0) to (1.center);
	\end{pgfonlayer}
\end{tikzpicture}
  \ \overset{\sim}{\longleftrightarrow}\ 
  \input{graymult-correct.tikz}
  \]
\end{remark}

\begin{remark}
  The map $s$ is called an antipode because equation~\eqref{eq:complementarity} is the antipode law for a Hopf algebra. If the spiders associated with a pair of mutually unbiased bases satisfy the other three Hopf algebra laws, they are called \textit{strongly complementary bases}, which are special mutually unbiased bases which always arise from finite abelian groups via Fourier transform~\cite{CKbook}.
\end{remark}

\section{Key Distribution Protocol}\label{sec:qkd}

%Todo: Need to give a name to Alice and Bob's Hilbert space, at some point before main Theorem (currently `A', could be H).

Alice chooses a random bit and encodes this bit as a qubit using either the Z $\whitedot$ or X $\graydot$ basis with equal probability. She sends the qubit to Bob. Independently, Bob chooses to perform either a Z or X basis measurement on the qubit he received, again with equal probability.
With probability 1/2 their choices agree and the bit is perfectly transmitted:
\[
\input{encode1.tikz} 
\qquad \qquad \qquad
\input{encode1a.tikz} 
\]
Otherwise, Bob receives a random bit, with no information conveyed:
\[
\input{encode2.tikz} 
\qquad \qquad \qquad
\input{encode3.tikz} 
\]
To agree on a shared key with average size $n$, Alice and Bob go through the previous routine~$4n$ times.
Then Alice and Bob announce the bases they used for encoding and measuring respectively and discard those bits where the bases disagree.  On average~$2n$ bits remain. To check for trouble, Alice randomly picks~$n$ of the remaining bits and announces their value to Bob. If there is any mismatch on these check-bits, Alice and Bob abort. If not, Alice and Bob are left with (on average) $n$ bits.

Suppose that an eavesdropper Eve intercepts a transmitted qubit in attempt to extract information. Eve does not know which basis Alice used to encode her bit, so let us first assume Eve adopts a naive strategy whereby she randomly decides to perform a non-demolition measurement with respect to Z or X.

Alice and Bob will ignore any bits for which their basis choices differ, so we need only consider the case where they match. If Alice and Bob both choose Z, and Eve chooses (correctly) to also measure in Z:
\begin{equation*}% \label{eq:Eve_requirements_exact}
\input{outcomes1.tikz}
\end{equation*}
she indeed receives a perfect copy of Alice's bit. On the other hand, if she guesses wrong, and measures X:
\begin{equation*}% \label{eq:Eve_requirements_exact}
\input{outcomes2.tikz}
\end{equation*}
Eve and Bob both receive a random bit, completely uncorrelated to Alice's original bit. Alice and Bob's communication will be similarly disrupted when they both measure X but Eve measures Z. In a longer run, these disruptions will be detected by Alice and Bob using the check-bits, giving away Eve's presence. 

Now let us consider the case where Eve can apply any operation to attempt to extract some information about Alice's bit. That is, show applies some quantum channel $\Phi \colon \mathcal{B}(\ahilb) \to \mathcal{B}(\ahilb \otimes \ehilb)$, where $\ahilb$ is the Hilbert space of the qubit and $\ehilb$ is that of some other system possessed by Eve:
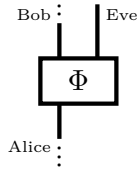
\ctikzfig{eve}

In order for Eve's intervention to remain undetected in either case where Alice and Bob's measurements agree, Eve's channel must satisfy the following equations:
\begin{equation}\label{eq:eve-secret}
\input{exact-req1-enc.tikz}\qquad \qquad
\input{exact-req2-enc.tikz}
\end{equation}
We now prove that, in this scenario, Eve cannot possibly extract any information. That is, her channel separates.

\begin{theorem} \label{thm:exact_security_QKD}
For any $\Phi$ satisfying \eqref{eq:eve-secret} for complementary ONBs \whitedot/\graydot we have:
\[
%\tikzfig{separates}
\input{separateslabels.tikz}
\]
for some state $\rho$ of $\mathcal{B}(\ehilb)$.
\end{theorem}

% Hence Eve gains no information about the transmitted qubit.

\begin{proof}
By precomposing measurement maps and postcomposing encoding maps, we find
\begin{equation} \label{eq:Eve_requirements_exact}
\input{exact-req1.tikz}\qquad \qquad
\input{exact-req2.tikz}
\end{equation}
By purifying $\Phi$, we can without loss of generality assume that Eve's map $\Phi$ is pure, say with $\Phi = \widehat \evepure$. The left-hand side of~\eqref{eq:Eve_requirements_exact} states that:
\begin{equation} \label{eq:addpurestate}
\input{proof1.tikz}
\end{equation}
for any normalized pure state $\psi$ of $\ehilb \otimes \ahilb$. By essential uniqueness of purification, we conclude that:
\begin{equation} \label{eq:unitaryappearsinproof}
\input{proof3.tikz}
\end{equation}
for some unitary $U$ on $\ahilb \otimes \ehilb \otimes \ahilb$. Hence:
\begin{equation} \label{eq:counitsontopinproof}
\input{proof4.tikz}
\end{equation}
Using the spider laws~\eqref{eq:spider1}, it follows that: 
\begin{equation} \label{eq:phaseproofstep}
\input{proof5.tikz}
\end{equation}
The same equations also hold for the gray spiders
    with the same reasoning
    starting with the right-hand side of~\eqref{eq:Eve_requirements_exact}.
But together these imply that $\evepure$ separates, since:
\[
\input{proof8p.tikz}
\]
\[
\input{proof9p.tikz}
\]
(where for a qubit we have $D = 2$). Hence $\Phi$ separates as well.
%\begin{widetext}
% \[
% \includegraphics[scale=0.2]{figures/pic10}
% \]
% \[
% \includegraphics[scale=0.2]{figures/pic10'}
% \]
%\end{widetext}
\end{proof}

The simple nature of the graphical proof makes clear several immediate generalizations. Firstly, the same protocol and proof may be applied apply whether Alice and Bob are sharing a qubit or a \emph{qudit} for arbitrary (finite) Hilbert space dimension $D$, on which Alice and Bob must simply choose any pair of complementary orthonormal bases.

\section{Eavesdroppers with memory}\label{sec:memory}

A very strong assumption in the previous derivation was that Eve performs the \textit{same} channel every time in attempts to extra information from Alice's string of bits. Suppose now that Eve may now vary her behavior depending on the qubits she has received previously. That is, she may now make use of a \emph{quantum memory} that persists between individual steps of the protocol. Her channel is now of the form $\Phi \colon \mathcal{B}(\ahilb \otimes \ehilb) \to \mathcal{B}(\ahilb \otimes \ehilb)$, with an extra input into which is passed the output from the previous intercepted qubit. She initially prepares her auxiliary system in some state $\rho$. Then Eve's procedure, during $n$ transmissions between Alice and Bob, amounts to the channel $\Phi' \colon \mathcal{B}(\ahilb^{\otimes n}) \to \mathcal{B}(\ahilb^{\otimes n} \otimes \ehilb)$ given by:
\begin{equation}% \label{eq:quantummemory}
%\includegraphics[scale=0.25]{figures/pic_mem3}
%\tikzfig{memory1'}
\input{memory1labels.tikz}
\end{equation}
(As Eve can keep a counter in~$E$,
    this setting is as general as if we would allow Eve
    a different~$\Phi_n\colon \mathcal B(H \otimes E) \to \mathcal B (H \otimes E)$
    for each transmitted qubit.)

As before, for Eve's interference to remain completely undetected,
    we must have that: 
%Again to remain undetected, $\Phi'$ must satisfy:
\begin{equation} \label{eq:quantummemoryreq}
\input{memory2half.tikz}
\end{equation}
along with analogous equations for all of the remaining $2n-1$ variations of basis and position. In particular the one for the first bit becomes:

%ALTERNATIVE ARGUMENTS ==========
% \textbf{First Argument}
% But already the requirement~\eqref{eq:quantummemoryreq} may be seen as a special case of our previous siutation. For this, simply note that the $n$-fold tensor of encoding maps in~\eqref{eq:quantummemoryreq} together forms a single encoding map for the family of spiders on $\ahilb^{\otimes n}$ corresponding ot the $n$-fold product basis $\{ \ket {i_1} \otimes... \otimes \ket {i_n} \} \subseteq \ahilb^n$ (and the same holds for the decoding maps). Moreover, whenever two bases are complementary, so are their $n$-fold tensor products. Hence we may apply Theorem~\ref{thm:exact_security_QKD} to conclude that $\Phi'$ separates, and so Eve again gains no information.

% \textbf{Alternative wording}: Indeed it is easy to see that, for any family of spiders on a Hilbert space $H$ corresponding to the ONB $\{ \ket i \} \subset H$ their tensor product:
% \begin{equation} %\label{eq:Eve_requirements_exact}
% \includegraphics[scale=0.25]{figures/pic_spidertensor}
% \end{equation}
% forms a family of spiders on $H \otimes H$, corresponding to the ONB  $\{ \ket i \otimes \ket j \} \subset H \otimes H$. Similarly, the same holds for $n$-fold tensor products. When two sets of spiders are complementary, their $n$-fold tensor products are again also.

%========================
%\textbf{Alternative Inductive argument}
\begin{equation*} %\label{eq:Eve_requirements_exact}
\input{memory3.tikz}
\end{equation*}
The same equation holds for the gray spiders also. From cancellativity of the tensor product, we conclude that $\Phi \circ (\id{} \otimes \rho)$ satisfies the requirements of Theorem~\ref{thm:exact_security_QKD}, and so separates as: %as $\id{} \otimes \sigma$
\begin{equation*} %\label{eq:Eve_requirements_exact}
\input{memory4.tikz}
\end{equation*}
for some state $\rho'$ of $\mathcal{B}(\ehilb)$. 
%, by Theorem~\ref{thm:exact_security_QKD}. 
Returning to the requirement~\eqref{eq:quantummemoryreq}, the same argument again with ${n-1}$ instances of $\Phi$ and $\rho'$ replacing $\rho$ shows that $\Phi \circ (\id{} \otimes \rho')$ itself separates. Repeating this argument inductively reveals that $\Phi$ itself separates, and thus Eve again receives no information.
%By induction, so does $\Phi$.
%Inductively,  and then inductively that so does $\Phi$.
%But then the requirement~\eqref{eq:quantummemoryreq} is just the same as before, now with $n-1$ instances of $\Phi$ and $\rho'$ replacing $\psi$. Hence, by the same argument, $\Phi \circ (\id{} \otimes \rho)$ separates, and inductively so does $\Phi$.

 %The $n$-fold encoding and decoding maps in~\eqref{eq:quantummemoryreq} are thus those correpsonding to the complementary 

% The whole process, with $n$-tranmissions between Alice and Bob, is then represented by the following map $A^{\otimes n} \to B^{\otimes n}$:
% \begin{equation} \label{eq:quantummemory}
% \includegraphics[scale=0.3]{figures/pic_mem}
% \end{equation}

\section{Noise-tolerance}\label{sec:noise}

The proofs we have just seen might be pleasing, but they do not give us \emph{a priori} confidence in the security of QKD: it might be the case that with only a minute disturbance, Eve might extract a lot of information. We will see in this section that this is not the case: the equational proof also leads to a polynomial bound on the distance of Eve's channel from a separable one by the amount of disturbance. The key is an approximate version of essential uniqueness of purifications due to Kretschmann, Schlingemann and Werner~\cite{kretschmann2008continuity,contstinespring}:
\begin{theorem}[Continuity of Stinespring Dilation] \cite[Theorem 1]{contstinespring} \label{thm:continousstinespring}
Let $V_1, V_2 \colon A \to B \otimes E$ be linear maps. Then:
% \begin{widetext}
% \begin{equation} \label{eq:continuitystinespring}
%     \inf_U \,\Biggl\|\  \tikzfig{v1-minus-v2}\  \Biggr\|^2_\infty \ \leq\  
%     \Biggl\| \ \tikzfig{v1-minus-v2-dbl} \  \Biggr\|_{\mathrm{cb}} \ \leq \ 
%     2\inf_U \,\Biggl\|\  \tikzfig{v1-minus-v2}\  \Biggr\|_\infty  
% %\includegraphics[scale=0.25]{figures/pic21}
% \end{equation}
% \end{widetext}

\begin{align*} \label{eq:continuitystinespring}
    &\inf_U \,\Biggl\|\  \input{v1-minus-v2.tikz}\   \Biggr\|^2_\infty \  \\ 
    &\qquad\leq\ \Biggl\| \ \input{v1-minus-v2-dbl.tikz} \  \Biggr\|_{\mathrm{cb}} \\ \ 
    &\qquad\leq \  2\inf_U \,\Biggl\|\  \input{v1-minus-v2.tikz}\  \Biggr\|_\infty  
\end{align*}
where the infima are taken over all unitaries $U \colon E \to E$.
% \[
% \includegraphics[scale=0.3]{figures/pic_14}
% \]
% Also we have
% \begin{equation} \label{eq:cont_second}
% \includegraphics[scale=0.3]{figures/pic_21}
% \end{equation}
%\[ Alternative Version
%\includegraphics[scale=0.3]{figures/pic_15}
%\]
\end{theorem}

% We will need to make use of two norms:
% \begin{itemize}
% \item for linear maps $V \colon H \to K$ between Hilbert spaces, we use the usual \emph{operator norm}: 
% \[
% ||V||_{\infty} = \underset{||x|| = 1}\sup ||V(x)||
% \]
% We often write $V \overset{\epsilon}{=} V'$ to mean that $|| V - V ||_{\infty} \leq \epsilon$.
% \item for completely positive maps $T \colon \mathcal{B}(H) \to \mathcal{B}(K)$ we use the \emph{completely bounded norm} \cite{kretschmann2008information}:
% \[
% ||T_1 - T_2||_{\cb} = \underset{n \in \mathbb{N}}\sup \  || \id_{\mathcal{B}(\mathbb{C}^n)} \otimes (T_1 - T_2) ||_{\infty}
% \]
% where for completely positive maps $T_1$, $T_2$, $||T_1 - T_2||{\infty}$ is defined using the norm on $\mathcal{B}(H)$ given by $|| a || = \Tr(\sqrt{a^*a})$.
% We will write $T \underset{cb}{\overset{\epsilon}{=}} T'$ to mean $||T - T'||_{\cb} \leq \epsilon$.
% \end{itemize}

Here~$\|\,\cdot\,\|_\infty$ denotes the usual \emph{operator norm}
and~${\|\,\cdot\,\|_\cb}$ the \emph{completely bounded norm}
of super operators which is defined via the sup-norm on super-operators:
\begin{equation*}
    \|V \|_\infty := \sup_{\|\psi\|\leq 1} \|V \psi\|
\qquad \|\Phi\|_\infty := \sup_{\| T \| \leq 1}  \| \Phi(T) \|_\infty
\end{equation*}
\begin{equation*}
    \|\Phi\|_\cb := \sup_{n \in \N} \| \id_{M_n}\otimes \Phi\|_\infty
\end{equation*}

%Sean: Alternative layout.
% \begin{align*}
%     \|V \|_\infty &:= \sup_{\|\psi\|\leq 1} \|V \psi\|
% \\ {\qquad }\|\Phi\|_\infty &:= \sup_{\| T \| \leq 1}  \| \Phi(T) \|_\infty
% \\    \|\Phi\|_\cb &:= \sup_{n \in \N} \| \id_{M_n}\otimes \Phi\|_\infty
% \end{align*}
% Here~$\|\,\cdot\,\|_\infty$ denotes the usual \emph{operator norm} or sup-norm:% or sup-norm:
% \begin{equation*}
%     \|V \|_\infty := \sup_{\|\psi\|\leq 1} \|V \psi\|
% \qquad \|\Phi\|_\infty := \sup_{\| T \| \leq 1}  \| \Phi(T) \|_\infty
% \end{equation*}
% and~${\|\,\cdot\,\|_\cb}$ the \emph{completely bounded norm}
% on super operators, which is defined via the sup-norm on super-operators: % on super-operators:
% \begin{equation*}
%     \|\Phi\|_\cb := \sup_{n \in \N} \| \id_{M_n}\otimes \Phi\|_\infty
% \end{equation*}

The operator norm (on operators) and cb-norm (on super operators)
satisfy the following rules.
\begin{equation*} %\label{eq:opnormrules}
    \| f \circ g \| \leq \| f \| \, \|g \|
    \quad \| f \otimes  g \| = \| f \| \, \|g \|
    \quad \| \id \| = 1
\end{equation*}
The regular sup-norm on super operators does not satisfy the middle rule. For brevity, we will write
\begin{align*}
    V \mathrel{\overset{\varepsilon}{=}} V'
&\ \Leftrightarrow\ 
\| V - V' \|_{\infty}\  \leq \ \varepsilon \\
T \mathrel{\underset{\cb}{\overset{\varepsilon}{=}}} T' &\ \Leftrightarrow \  
\|T - T'\|_{\cb} \ \leq \ \varepsilon.
\end{align*}

The compositional rules for the operator norm allow us to lift approximate equations between sub-diagrams to those between full diagrams, using the rule:
\begin{equation} \label{eq:diagreasoning}
\begin{tikzpicture}
	\begin{pgfonlayer}{nodelayer}
		\node [style=small box] (0) at (0, 0) {$V$};
		\node [style=none] (1) at (0, -1.25) {};
		\node [style=none] (2) at (0, 1.25) {};
		\node [style=right label] (3) at (0.25, -1) {};
		\node [style=right label] (4) at (0.25, 1) {};
	\end{pgfonlayer}
	\begin{pgfonlayer}{edgelayer}
		\draw (1.center) to (0);
		\draw (0) to (2.center);
	\end{pgfonlayer}
\end{tikzpicture}
\overset{\varepsilon}{=}
\begin{tikzpicture}
	\begin{pgfonlayer}{nodelayer}
		\node [style=small box] (0) at (0, 0) {$W$};
		\node [style=none] (1) at (0, -1.25) {};
		\node [style=none] (2) at (0, 1.25) {};
		\node [style=right label] (3) at (0.25, -1) {};
		\node [style=right label] (4) at (0.25, 1) {};
	\end{pgfonlayer}
	\begin{pgfonlayer}{edgelayer}
		\draw (1.center) to (0);
		\draw (0) to (2.center);
	\end{pgfonlayer}
\end{tikzpicture}
\implies
\input{diag-reason.tikz}
\end{equation}
%\begin{equation} \label{eq:diagreasoning}
%\includegraphics[scale=0.25]{figures/pic_diagreason}
%\end{equation}
%We call this technique \emph{$\varepsilon$-bounded diagrammatic reasoning}.

%Using these compositional rules for the operator norm, given a diagram and a subdiagram equal to some other up to $\varepsilon$ (in the operator norm), we can deduce a bounded equation between 
%an equation between subdiagrams up to some $\varepsilon$ (in the operator norm), we can deduce an equation between the whole diagrams 
%Thanks to these, given a diagram and a subdiagram equal to some other map up to $\varepsilon$ (in the operator norm), we can deduce 
%an equation between subdiagrams up to some $\varepsilon$ (in the operator norm), we can deduce an equation between the whole diagrams 
%These allow us to replace strict equations making use of the rules~\eqref{eq:opnormrules} for the operator norm, which allow us to replace 

\begin{theorem}[Noise Tolerance] \label{thm:noise_tolerance}
There is a constant ${\noiseconst}$, depending only on the dimension of Alice and Bob's system, such that whenever
%Suppose that Eve is detectable only up to noise $\varepsilon$:
\begin{equation} \label{eq:Evenoise}
\input{approx-req1.tikz}\qquad\qquad
\input{approx-req2.tikz}
\end{equation}
we have that  
\begin{equation}
\label{eq:conclusionnoise}
\input{separates-approx.tikz}
\end{equation}
for some state $\rho$ of $\mathcal{B}(\ehilb)$.
\end{theorem}

\begin{proof}
From the upper bound of Theorem~\ref{thm:continousstinespring}, it suffices to prove that there is such a constant $\noiseconst$ for which any purification $\widehat \evepure$ of Eve's channel satisfies:
\begin{equation} \label{eq:sep_noise}
\input{separates-approx-pure.tikz}
\end{equation}
This follows from the proof of Theorem~\ref{thm:exact_security_QKD}, by repeatedly applying the rule~\eqref{eq:diagreasoning} to replace each strict equation in the proof by an approximate one. At each step one only picks up linear factors from the norms of other maps featuring in each diagram. Inspecting the proof, one can see that these are independent of Eve's system or channel.  
 %Simply use the rule~\eqref{eq:diagreasoning} for the operator norm to lift approximate equations between subdiagrams to larger diagrams. 
%up to linear factors given by the norm of the remaining diagram. Inspecting the proof, one can see that these factors are independent of Eve's system or channel. 
 % $\|\,\cdot\,\|_{\infty}$.

In more detail, let $\widehat \evepure$ be a purification of Eve's channel. Then the left hand side of~\eqref{eq:Evenoise} says precisely that the two sides of~\eqref{eq:addpurestate} are within $\varepsilon$ in the completely bounded norm, for any such pure state $\psi$. By Theorem~\ref{thm:continousstinespring}, there is then a unitary $U$ such that the two sides of~\eqref{eq:unitaryappearsinproof} are within $2 \sqrt{\varepsilon}$. Applying two copies of $\whitedot^1_0$ as in \eqref{eq:counitsontopinproof}, we obtain equation~\eqref{eq:counitsontopinproof} up to $2 \sqrt{\varepsilon}||\whitedot^1_0||^2$. 
This ensures that~\eqref{eq:phaseproofstep} holds up to a constant factor in $\sqrt{\varepsilon}$ dependent only on the norms of the $\whitedot^n_m$ maps. Similarly so does the same equation for the gray spider.

The final steps of the proof now follow just as before. Again at each step we simply replace a strict equation by one up to a constant factor in $\sqrt\varepsilon$, dependent only on the white and gray spiders, using the rule~\eqref{eq:diagreasoning}. %for $||\,\cdot\,||_{\infty}$.
%This ensures that~\eqref{eq:phaseproofstep} holds up to a constant factor in $\sqrt{\varepsilon}$ dependent only on the norms of the white spider maps. Similarly so does the same equation for the grey spider. Finally, we may copy the final proof steps as before, in each step picking up only constant factors given by the norms of white and grey spider morphisms.
% to give~\eqref{eq:sep_noise}, in each step picking up only 
\end{proof}

Taking the distance of Eve's channel from a separable channel as a measure of how much information she can extra from Alice's bit, we see that, as expected, the amount of information is bounded by how much Eve disturbs the communication between Alice and Bob.

We conclude this section with a few caveats:
\begin{enumerate}
\item
    The bound applies if the systems of Alice, Bob and Eve can be modeled by finite-dimensional Hilbert spaces and their interaction as a completely positive map between their tensor products.  For other applications these assumptions are untenable, see e.g.~\cite{yngvason2005role} and \cite{shaji2005s}. There is, however, no clear indication that for quantum information arguments these assumptions are invalid.
Nonetheless it's of interest how far these arguments generalize. Interactions between arbitrary (possibly infinite-dimensional) von Neumann algebras admit a Stinespring-like dilation \cite{westerbaan2017paschke}, but no continuity result is known.
\item
    Like the proof in Section~\ref{sec:qkd}, the proof above assumes Eve performs the same operation every time. This should be extended to incorporate memory as in Section~\ref{sec:memory}.
\item
    The random sampling (check-bits) on its own does not guarantee the bound \eqref{eq:Evenoise}. Renner solves the analogous problem with some effort by showing the resulting combined state of Alice, Bob and Eve may be assumed to be approximately symmetric under permutations of the bits in the run~\cite{renner}. A similar method may apply here.
\end{enumerate}

\section{Outlook}
We saw how to derive a bound on the security of QKD using the diagrammatic behavior of complementary observables and continuity of Stinespring.  Whether the protocol ensures the assumed bound \eqref{eq:Evenoise} \emph{and}
how the conclusion \eqref{eq:conclusionnoise} is related to the more common quantities like mutual information and the secret-key rate, we leave open to future research.

In closing we note that, in the current literature, diagrammatic arguments have typically only been used for exact reasoning about quantum processes. Our proof of Theorem~\ref{thm:noise_tolerance} suggests a  general strategy of using the rule~\eqref{eq:diagreasoning} to extend such exact diagrammatic arguments to approximate ones. We call this technique \emph{$\varepsilon$-bounded diagrammatic reasoning}. It should be applicable to many further quantum protocols, allowing the intuitive diagrammatic approach to quantum information theory developed in~\cite{CKbook} to be used to make physically reasonable and robust arguments, such as required for security protocols. 

\paragraph{Acknowledgments}
We thank Hans Maassen, Sam Staton, Bram Westerbaan and John van de Wetering for insightful discussion. We would also like to thank the anonymous QPL referees for their useful feedback on a short abstract summarising this work.

This work is supported by the ERC under the European Union's Seventh Framework Programme (FP7/2007-2013) / ERC grant n\textsuperscript{o} 320571 and an EPSRC Studentship OUCL/2014/SET.

%one to lift the

% Previously, the diagrammatic approach

% In closing, we note that the proof technique~\eqref{eq:diagreasoning} appears to be a simple and powerful one, more general. 

% It straightforwardly allowed us to lift our exact diagrammatic argument to a more physically reasonable approximate one, relevant for security purposes. 
% It suggests it should be applicable to more security protocols in the future. 

% We call this technique \emph{$\varepsilon$-bounded diagrammatic reasoning}.

% Combine high level intuitive nature of the diagrammatic approach to quantum information developed in~\cite{CKbook} with the physically reasonable and robustness required for security protocols. 
% Hasn't been considered previously, diagrammatics only used for exactness.  

% It allows us to exact diagrammatic arguments to more physically ones.

% Future. 

%Todo: mention this in a conclusion anywhere?
% \begin{remark} We also did not use commutativity of the spiders. Can we generalise to von Neumann measurements?
% \end{remark}
%Todo: Discuss $\varepsilon$-bounded diagrammatic reasoning here. 

\bibliographystyle{abbrvnat}
\bibliography{main}

\end{document}